\documentclass[11pt]{article}
\usepackage{amsmath}
\usepackage{amsfonts}
\usepackage{amssymb}
\usepackage{dsfont}
\usepackage{mathrsfs}                  
\usepackage{amsthm}
\usepackage{enumitem}
\usepackage{amscd}
\usepackage{float}
\usepackage{xcolor}
\usepackage[hyperfootnotes=false]{hyperref}
\usepackage{mathtools}
\usepackage{geometry}
\usepackage{bbm}
\usepackage{bm}
\usepackage{bbm}
\usepackage{graphicx}
\usepackage[numbers,sort]{natbib}

\usepackage{subfigure}

\numberwithin{equation}{section}
\theoremstyle{plain}
\newtheorem{theorem}{Theorem}[section]
\newtheorem{proposition}[theorem]{Proposition}

\newtheorem*{lemma*}{Lemma}
\newtheorem{corollary}[theorem]{Corollary}
\newtheorem{assumption}{Assumption}

\theoremstyle{definition}
\newtheorem{definition}[theorem]{Definition}

\theoremstyle{remark}
\newtheorem{remark}[theorem]{Remark}

\DeclareMathOperator*{\esssup}{ess\,sup}
\DeclareMathOperator*{\essinf}{ess\,inf}

\newcommand{\E}{\mathbb{E}}
\newcommand{\R}{\mathbb{R}}

\newcommand{\N}{\mathbb{N}}

\newcommand*\samethanks[1][\value{footnote}]{\footnotemark[#1]}

\renewcommand{\P}{\textbf{P}}

\newcommand{\Fc}{\mathcal{F}}

\newcommand{\Yc}{\mathcal{Y}}

\newcommand{\activF}{\sigma}

\def \bE {{\mathbb{E}}}
\def \bF {{\mathbb{F}}}
\def \bN {{\mathbb{N}}}
\def \bR {{\mathbb{R}}}

\def \cA {{\mathcal{A}}}
\def \cB {{\mathcal{B}}}

\def \cF {{\mathcal{F}}}

\def \cP {{\mathcal{P}}}

\def \cR {{\mathcal{R}}}

\def \cU {{\mathcal{U}}}
\def \cV {{\mathcal{V}}}

\def \cY {{\mathcal{Y}}}

\def \bP {{\textbf{P}}}

\begin{document}
	
\title{Neural network approximation for superhedging prices}

\author{Francesca Biagini\thanks{Workgroup Financial and Insurance Mathematics, Department of Mathematics, Ludwig-Maximilians Universit{\"a}t, Theresienstra{\ss}e 39, 80333 Munich, Germany. Emails: biagini@math.lmu.de, gonon@math.lmu.de, reitsam@math.lmu.de.} \and Lukas Gonon\samethanks[1] \and Thomas Reitsam\samethanks[1] \thanks{The financial support of the Verein zur Versicherungswissenschaft M\"unchen e.V. is gratefully acknowledged} }

\maketitle

\abstract{\textit{This article examines neural network-based approximations for the superhedging price process of a contingent claim  in a discrete time market model. First we prove that the $\alpha$-quantile hedging price converges to the superhedging price at time $0$ for $\alpha$ tending to $1$, and show that the $\alpha$-quantile hedging price can be approximated by a neural network-based price. This provides a neural network-based approximation for the superhedging price at time $0$ and also the superhedging strategy up to maturity. To obtain the superhedging price process for $t>0$, by using the Doob decomposition it is sufficient to determine the process of consumption. We show that it can be approximated by the essential supremum over a set of neural networks. Finally, we present numerical results.}}\bigskip

\noindent
\textbf{Keywords:}\begin{small} Deep learning; Superhedging; Quantile hedging\end{small}
\smallbreak\noindent
\smallbreak\noindent
\bigskip

\section{Introduction}
\label{sec:introduction}

In this paper we study neural network approximations for the superhedging price process for a contingent claim in discrete time. \\
Superhedging was first introduced in \cite{quenez} and then thoroughly studied in various settings and market models. It is impossible to cover the complete literature here, but we name just a few references. For instance, in continuous time, for general c\`adl\`ag processes we mention \cite{kramkovduality}, for robust superhedging \cite{nutz2015robust}, \cite{touzi2014martingale}, for pathwise superhedging on prediction sets \cite{bartl2020pathwise}, \cite{bartl2019duality}, or for superhedging under proportional transaction costs \cite{campischachermayer}, \cite{cvitanic}, \cite{kabanovlast}, \cite{schachermayersuperreplication}, \cite{soner1995there}. Also in discrete time there are various studies in the literature, like the standard case \cite{follmerschied}, robust superhedging \cite{carassus2019robust}, \cite{obloj2021robust}, superhedging under volatility uncertainty \cite{nutz2012superhedging}, or model-free superhedging \cite{burzoni2017model}. The superhedging price provides an opportunity to secure a claim, but it may be too high or reduce the probability to profit from the option. In order to solve this problem, quantile hedging was introduced in \cite{foellmer99}, where the authors propose to either fix the initial capital and maximize the probability of superhedging with this capital or fix a probability of superhedging and minimize the required capital. In this way a trader can determine the desired trade-off between costs and risk. \\
In certain situations it is possible to calculate explicitly or recursively the superhedging or quantile hedging price, see e.g.\ \cite{carassus2007class}, but in general incomplete markets it may be complicated to determine superhedging prices or quantile hedging prices. In this article we investigate neural network-based approximations for quantile- and superhedging prices. Neural network-based methods have been recently introduced in financial mathematics, for instance for hedging derivatives, see \cite{buehler2018}, determining stopping times, see \cite{becker2019deep}, or calibration of stochastic volatility models, see \cite{cuchiero2020generative}, and many more. For an overview of applications of machine learning to hedging and option pricing we refer to \cite{ruf2020neural} and the references therein. \\
This paper contributes to the literature on hedging in discrete time market models in several ways. First, we prove that the $\alpha$-quantile hedging price converges to the superhedging price for $\alpha$ tending to $1$. Further, we show that it is feasible to approximate the $\alpha$-quantile hedging and thus also the superhedging price for $t=0$ by neural networks. We extend our machine learning approach also to approximate the superhedging price process for $t>0$. By the first step we obtain an approximation for the superhedging strategy on the whole interval up to maturity. By using the uniform Doob decomposition, see \cite{follmerschied}, we then only need to approximate the process of consumption $B$ to generate the superhedging price process. We prove that $B$ can be obtained as the the essential supremum over a set of neural networks. Finally, we present and discuss numerical results for the proposed neural network methods. \\
The paper is organized as follows. In Section \ref{sec:prelim}, we present the discrete time market model of \cite{follmerschied} and recall essential definitions and results on superhedging. Section \ref{sec:price0} contains the study of the superhedging price for $t=0$. More specifically, in Section \ref{sec:quantilehedging} we prove in Theorem \ref{thm:convergencesuccessset} that the $\alpha$-quantile hedging price converges to the superhedging price as $\alpha$ tends to $1$. We also present a similar result in Corollary \ref{cor:successratioconvergence} in Section \ref{sec:successratio}, where $\alpha$-quantile hedging is given in terms of success ratios. In Section \ref{sec:NNprice0} we show in Theorem \ref{thm:universalapprox2} that the superhedging price can be approximated by neural networks. This concludes the approximation for $t=0$. Then, we consider the case for $t>0$ in Section \ref{sec:superhedgingt}. In Section \ref{sec:uniformDoob}, we explain how the uniform Doob decomposition can be used to approximate the superhedging price process. In that account, we prove an explicit representation of the process of consumption, see Proposition \ref{prop:optimizationB}. Proposition \ref{prop:NNconsumption} and Theorem \ref{thm:NNgeneralt} show that the process of consumption and thus the superhedging price process can be approximated by neural networks. The numerical results are presented in Section \ref{sec:numericalresults}. The section is divided in the case $t=0$, see Section \ref{sec:case0}, and $t>0$, see Section \ref{sec:caset}. We present details on the algorithm and the implementation. Appendix \ref{sec:appendixNN} contains a version of the universal approximation theorem, derived from \cite{hornik1991}.

\section{Preliminaries}
\label{sec:prelim}
In this section we introduce the discrete time financial market model from \cite{follmerschied} and recall some basic notions on superhedging. \\
Consider a finite time horizon $T>0$. Let $(\Omega,\cF,\bP)$ be a probability space endowed with a filtration $\bF:=(\cF_t)_{t=0,1,\dots,T}$. We assume $\cF_t=\sigma(Y_0,\dots,Y_t)$ for $t=0,\dots,T$ and for some $\bR^m$-valued process $Y=(Y_t)_{t=0,\dots,T}$ for some $m\in\bN$, and write $\cY_t=(Y_0,\dots,Y_t)$ for $t\geq 0$. Further, we suppose that $\cF=\cF_T$ and that $Y_0$ is constant $\bP$-a.s. Then $\cF_0=\{\emptyset,\Omega\}$.\\
In our market model on $(\Omega,\cF,\bF,\bP)$ the asset prices are modeled by a non-negative, adapted, stochastic process 
\[
\bar S=(S^0,S)=(S_t^0,S_t^1,\dots,S_t^d)_{t=0,1,\dots,T},
\]
with $d\geq 1$, $d\in\bN$. In particular, $m\geq d$. Further, we assume that 
\[
S_t^0>0 \quad\bP\text{-a.s. for all }t=0,1,\dots,T,
\]
and define $S^0=(S_t^0)_{t=0,1,\dots,T}$ to be the num\'eraire. The discounted price process $\bar X=(X^0,X)=(X_t^0,X_t^1,\dots,X_t^d)_{t=0,1,\dots,T}$ is given by
\[
X_t^i:=\frac{S_t^i}{S_t^0},\quad t=0,1,\dots,T,\ i=0,\dots,d.
\]
A probability measure $\bP^*$ is called an equivalent martingale measure if $\bP^*$ is equivalent to $\bP$ and $X$ is a $\bP^*$-martingale. We denote by $\cP$ the set of all equivalent martingale measures for $X$ and assume $\cP\neq \emptyset$. By Theorem 5.16 of \cite{follmerschied} this is equivalent to the market model being arbitrage-free.
\begin{definition}
	A trading strategy is a predictable $\bR^{d+1}$-valued process 
	\[
	\bar\xi=(\xi^0,\xi)=(\xi_t^0,\xi_t^1,\dots,\xi_t^d)_{t=1,\dots,T}.
	\]
	The (discounted) value process $V=(V_t)_{t=0,\dots,T}$ associated with a trading strategy $\bar\xi$ is given by
	\[
	V_0:=\bar\xi_1\cdot \bar X_0\quad\text{and}\quad V_t:=\bar\xi_t\cdot \bar X_t\quad \text{for }t=1,\dots,T.
	\]
	A trading strategy $\bar\xi$ is called self-financing if
	\[
	\bar\xi_t\cdot\bar S_t=\bar\xi_{t+1}\cdot\bar S_t\quad\text{for }t=1,\dots,T-1.
	\]
	A self-financing trading strategy is called an \emph{admissible strategy} if its value process satisfies $V_T\geq 0$.\\
	By $\cA$ we denote the set of all admissible strategies $\bar\xi$ and by $\cV$ the associated value processes, i.e.,
	\[
	\cV:=\left\lbrace V=(V_t)_{t=0,1,\dots,T}:\, V_t=\bar\xi_t\cdot\bar X_t\text{ for }t=0,\dots,T,\text{ and } \bar\xi\in\cA\right\rbrace
	\]
\end{definition}\noindent
By Proposition 5.7 of \cite{follmerschied}, a trading strategy $\bar\xi$ is self-financing if and only if
	\[
	V_t=V_0+\sum_{k=1}^t\xi_k\cdot(X_k-X_{k-1})\quad\text{for all }t=0,\dots,T,
	\]
	with $V_0:=\bar\xi_1\cdot \bar X_0$. In particular, given a $\bR^d$-valued predictable process $\xi$ and $V_0\in\bR$, the pair $(V_0,\xi)$ uniquely defines a self-financing strategy. 

\begin{remark}
	By Theorem 5.14 of \cite{follmerschied}, $V_T\geq 0$ $\bP$-a.s. implies that $V_t\geq 0$ $\bP$-a.s. for all $t=0,\dots,T$, where $V$ denotes the value process of a self-financing strategy. More precisely, Theorem 5.14 of \cite{follmerschied} guarantees that if $\bar\xi$ is a self-financing strategy and its value process $V$ satisfies $V_T\geq 0$, then $V$ is a $\bP^*$-martingale for any $\bP^*\in\cP$. In particular, in the proof the martingale property of $X$ and Proposition 5.7 of \cite{follmerschied} is used successively in the following way:
	\[
	\bE^*[V_T\mid\cF_t]=\bE^*[V_{T-1}+\xi_T\cdot(X_T-X_{T-1})\mid\cF_{T-1}]=V_{T-1}+\xi_T\cdot\bE^*[X_T-X_{T-1}\mid\cF_{T-1}]=V_{T-1}.
	\]
\end{remark}
A discounted European contingent claim is represented by a non-negative, $\cF_T$-measurable random variable $H$ such that
\[
\sup_{\bP^*\in\cP}\bE^*[H]<\infty.
\]

\begin{definition}
\label{def:superhedging}
	Let $H$ be a European contingent claim. A self-financing trading strategy $\bar\xi$ whose value process $V$ satisfies 
	\[
	V_T\geq H\quad\bP\text{-a.s.}
	\]
	is called a \emph{superhedging strategy} for $H$. In particular, any superhedging strategy is admissible since $H\geq 0$ by definition. 
\end{definition}\noindent
The upper Snell envelope for a discounted European claim $H$ is defined by
\[
U_t^\uparrow(H)=U_t^\uparrow:=\esssup_{\bP^*\in\cP}\bE^*[H\mid\cF_t],\quad\text{for }t=0,1,\dots,T.
\]
\begin{corollary}[Corollary 7.3, Theorem 7.5, Corollary 7.15, \cite{follmerschied}]
	\label{cor:resultsfollmerschied}
	The process
	\[
	\left(\esssup_{\bP^*\in\cP}\bE^*[H\mid\cF_t]\right)_{t=0,1,\dots,T},
	\]
	is the smallest $\cP$-supermartingale whose terminal value dominates $H$. Furthermore, there exists an adapted increasing process $B=(B_t)_{t=0,\dots,T}$ with $B_0=0$ and a $d$-dimensional predictable process $\xi=(\xi_t)_{t=1,\dots,T}$ such that
	\begin{equation}
		\label{eq:uniformdoob}
		\esssup_{\bP^*\in\cP}\bE^*[H\mid\cF_t]=\sup_{\bP^*\in\cP}\bE^*[H]+\sum_{k=1}^t\xi_k\cdot(X_k-X_{k-1})-B_t\quad\bP\text{-a.s. for all }t=0,\dots,T.
	\end{equation}
	Moreover, $\esssup_{\bP^*\in\cP}\bE^*[H\mid\cF_t]=\essinf\cU_t$ and
	\begin{equation}
		\label{eq:uniformdoobdominates}
		\esssup_{\bP^*\in\cP}\bE^*[H\mid\cF_t]+\sum_{k=t+1}^T\xi_k\cdot(X_k-X_{k-1})\geq H,\quad\text{for all }t=0,\dots,T.
	\end{equation}
\end{corollary}\noindent
The process $B$ in \eqref{eq:uniformdoob} is sometimes called process of consumption, see \cite{kramkovduality}. Equations \eqref{eq:uniformdoob} and \eqref{eq:uniformdoobdominates} yield
\begin{equation}
	\label{eq:optimization}
	\sup_{\bP^*\in\cP}\bE^*[H]+\sum_{k=1}^T\xi_k\cdot(X_k-X_{k-1})-H\geq B_t\geq B_{t-1}\geq 0\quad\text{for all }t=1,\dots,T.
\end{equation}
Set
\begin{equation}
	\label{eq:superhedgingprice}
	\cU_t:=\left\{\tilde U_t\in L^0(\Omega,\cF_t,\bP):\exists \tilde\xi\text{ pred. s.t. }\tilde U_t+\sum_{k=t+1}^T \tilde\xi_k\cdot(X_k-X_{k-1})\geq H\quad\bP\text{-a.s.}\right\}.
\end{equation}

\begin{corollary}[Corollary 7.18, \cite{follmerschied}]
Suppose $H$ is a discounted European claim with
\[
\sup_{\bP^*\in\cP}\bE^*[H]<\infty.
\]
Then 
\[
\esssup_{\bP^*\in\cP}\bE^*[H\mid\cF_t]=\essinf\cU_t(H).
\]
\end{corollary}\noindent
Corollary 7.18 of \cite{follmerschied} and \eqref{eq:uniformdoobdominates} guarantee that $U_t^\uparrow$ is the minimal amount needed at time $t$ to start a superhedging strategy and thus there exists a predictable process $\xi$ such that
\[
\esssup_{\bP^*\in\cP}\bE^*[H\mid\cF_t]+\sum_{k=t+1}^T\xi_k\cdot (X_k-X_{k-1})\geq H.
\]
Further, $U_0^\uparrow$ is called the \emph{superhedging price} at time $t=0$ of $H$ and coincides with the upper bound of the set of arbitrage-free prices.

\section{Superhedging price for $t=0$}
\label{sec:price0}
In this section we approximate the superhedging price for $t=0$ in two steps. In the first part, we introduce the theory of quantile hedging, see \cite{foellmer99}. In Theorem \ref{thm:convergencesuccessset} we prove that the quantile hedging price for $\alpha\in(0,1)$ converges to the superhedging price as $\alpha$ tends to $1$. Analogously, in Corollary \ref{cor:successratioconvergence} we prove that for $\alpha$ tending to $1$ also the success ratios for $\alpha\in (0,1)$ converge to the superhedging price.\\
In the second part, we prove in Theorem \ref{thm:universalapprox2} that the superhedging price and the associated strategies can be approximated by neural networks. 

\subsection{Quantile hedging}
\label{sec:quantilehedging}
\subsubsection{Success sets}
\label{sec:successsets}

In incomplete markets perfect replication of a contingent claim may not be possible. Superhedging offers an alternative hedging method but it presents two main disadvantages. From one hand the superhedging strategy not only reduces the risk but also the possibility to profit. On the other hand, the superhedging price may result to be too high.\\
Quantile hedging was proposed for the first time in \cite{foellmer99} to address these problems. Fix $\alpha\in (0,1)$. Given probability of success $\alpha\in (0,1)$ we consider the minimization problem
\begin{align}
	\nonumber
	\inf \cU_0^\alpha:=\inf\{u\in\bR:\ &\exists \xi=(\xi_t)_{t=1,\dots,T}\text{ predictable process with values in }\bR^d\text{ such that }\\
	\label{eq:quantilehedging}
	&(u,\xi)\text{ is admissible and }\bP\left(u+\sum_{k=1}^T\xi_k\cdot(X_k-X_{k-1})\geq H\right)\geq \alpha\}.
\end{align}
Here $1-\alpha$ is called the shortfall probability. Quantile hedging may be considered as a dynamic version of the value at risk concept.\\
For an admissible strategy $(u,\xi)$ with associated value process $V,$ we call
\[
\{V_T\geq H\}
\]
the \emph{success set}.
\begin{remark}
\label{rem:admissiblestrategies}
Note that in \eqref{eq:quantilehedging} we need to require that $(u,\xi)$ is admissible since this is not automatically implied by the definition of quantile hedging as in the case of superhedging strategies in Definition \ref{def:superhedging}.
\end{remark}\noindent
Proposition \ref{prop:alternativepropblem} below provides an equivalent formulation of the quantile hedging \eqref{eq:quantilehedging}, see also \cite{foellmer99}.

\begin{proposition}
\label{prop:alternativepropblem}
	Fix $\alpha\in (0,1)$. Then
	\[
	\inf\cU_0^\alpha=\inf\left\{\sup_{\bP^*\in\cP}\bE^*[H\mathds{1}_{A}]:\ A\in\cF_T, \ \bP(A)\geq \alpha\right\}.
	\]
\end{proposition}

\begin{proof}
	$``\leq"$: Take $A\in\cF_T$ such that $\bP(A)\geq \alpha$. We prove that 
	\begin{equation}
	\label{eq:alternativproblemproof1}
		\sup_{\bP^*\in\cP}\bE^*[H\mathds{1}_A]\in \left\{u \in\bR:\exists \xi\text{ adm. s.t. }\bP\left(u+\sum_{k=1}^T\xi_k\cdot(X_k-X_{k-1})\geq H\right)\geq \alpha\right\}.
	\end{equation}
	By the well-known superhedging duality, see Theorem 7.13 of \cite{follmerschied}, we have that 
	\[
		\sup_{\bP^*\in\cP}\bE^*[H\mathds{1}_A]=\inf \left\{u \in\bR:\exists \xi\text{ pred. s.t. }u+\sum_{k=1}^T\xi_k\cdot(X_k-X_{k-1})\geq H\mathds{1}_A\ \bP\text{-a.s.}\right\},
	\]
	and that there exists a superhedging strategy $\hat\xi$ for $H\mathds{1}_A$ with initial value $\sup_{\bP^*\in\cP}\bE^*[H\mathds{1}_A]$, i.e.,
	\begin{equation}
	\label{eq:superhedgingonA}
	\sup_{\bP^*\in\cP}\bE^*[H\mathds{1}_A]+\sum_{k=1}^T\hat\xi_k\cdot(X_k-X_{k-1})\geq H\mathds{1}_A\geq 0\ \bP\text{-a.s.}
	\end{equation}
	In particular, by \eqref{eq:superhedgingonA} we get for $\hat\xi$ that
	\[
	\bP\left(\sup_{\bP^*\in\cP}\bE^*[H\mathds{1}_A]+\sum_{k=1}^T\hat\xi_k\cdot(X_k-X_{k-1})\geq H\right)\geq \bP(A)	\geq \alpha.
	\]
	This implies \eqref{eq:alternativproblemproof1} and hence 
	\[
		\inf\cU_0^\alpha\leq \inf\left\{\sup_{\bP^*\in\cP}\bE^*[H\mathds{1}_{A}]:\ A\in\cF_T, \ \bP(A)\geq \alpha\right\}.
	\]
	$``\geq"$: Take $\tilde u\in \cU_0^\alpha$ and denote by $\tilde\xi=(\tilde\xi_k)_{k=1}^T$ the corresponding strategy such that 
	\[
	\bP\left(\tilde u+\sum_{k=1}^T\tilde\xi_k\cdot (X_k-X_{k-1})\geq H\right)\geq \alpha.
	\]
	Define the set $\tilde A$ by
	\[
		\tilde A:=\left\{\omega\in\Omega:\tilde u+\sum_{k=1}^T\tilde\xi_k(\omega)\cdot(X_k(\omega)-X_{k-1}(\omega))\geq H(\omega)\right\}.
	\]
	Clearly $\tilde A\in\cF_T$ and $\bP(\tilde A)\geq \alpha$. By construction we have that 
	\[
	\left(\tilde u+\sum_{k=1}^T\tilde\xi_k\cdot(X_k-X_{k-1})\right)\mathds{1}_{\tilde A}\geq H\mathds{1}_{\tilde{A}}\quad \bP\text{-a.s.}
	\]
	and because $\tilde\xi$ is assumed to be admissible, we have
	\[
	\left(\tilde u+\sum_{k=1}^T\tilde\xi_k\cdot(X_k-X_{k-1})\right)\mathds{1}_{\tilde A^c}\geq 0\quad \bP\text{-a.s.}
	\]
	In particular, $\tilde u\in \cU_0(H\mathds{1}_{\tilde A})$ and by Theorem 7.13 of \cite{follmerschied} we obtain
	\begin{equation}
	\label{eq:tildeaconstruction}
		\tilde u\geq \sup_{\bP^*\in\cP}\bE^*[H\mathds{1}_{\tilde A}]\in \left\{\sup_{\bP^*\in\cP}\bE^*[H\mathds{1}_{A}]:\ A\in\cF_T, \ \bP(A)\geq \alpha\right\}.
	\end{equation}
	That is, for an arbitrary $\tilde u\in \cU_0^\alpha$ we have constructed a set $\tilde A$ such that \eqref{eq:tildeaconstruction} holds. Therefore,
	\[
		\inf\cU_0^\alpha\geq \inf\left\{\sup_{\bP^*\in\cP}\bE^*[H\mathds{1}_{A}]:\ A\in\cF_T, \ \bP(A)\geq \alpha\right\}.
	\]
\end{proof}\noindent
Corollary 7.15 of \cite{follmerschied} guarantees that there exists a superhedging strategy with initial value $\inf\cU_0$. In contrast, there might be no explicit solution to the quantile hedging approach \eqref{eq:quantilehedging}. If a solution to the quantile hedging approach exists, then Proposition \ref{prop:alternativepropblem} states that it is given by the solution of the classical hedging formulation for the knockout option $H\mathds{1}_A$ for some suitable $A\in\cF_T$. However, such a set $A\in\cF_T$ does not always exist. In particular, quantile hedging does not always admit an explicit solution in general. The Neyman-Pearson lemma suggests to consider so-called success ratios instead of success sets. We will briefly discuss success ratios below. For further information we refer the interested reader to \cite{foellmer99}.\\
We now show that the superhedging price $\inf\cU_0$, can be approximated by the quantile hedging price $\inf\cU_0^\alpha$ for $\alpha$ tending to $1$. 
\begin{definition}
	For $\alpha\in (0,1)$ we define 
	\[
	\cF^\alpha:=\left\{A\in\cF_T:\bP(A)\geq \alpha\right\}.
	\]
\end{definition}

\begin{theorem}
	\label{thm:convergencesuccessset}
	The $\alpha$-quantile hedging price converges to the superhedging price as $\alpha$ tends to $1$, i.e.,
	\[
	\inf\cU_0^\alpha\xrightarrow{\alpha\uparrow 1}\inf\cU_0.
	\]
\end{theorem}

\begin{proof}
	We first note that, using Proposition \ref{prop:alternativepropblem} it suffices to prove
	\[
	\inf_{A\in\cF^\alpha}\sup_{\bP^*\in\cP}\bE^*[H\mathds{1}_A]\xrightarrow{\alpha\uparrow 1}\sup_{\bP^*\in\cP}\bE^*[H].
	\]
	Let $(\alpha_n)_{n\in\bN}\subset (0,1)$ be an increasing sequence such that $\alpha_n$ converges to $1$ as $n$ tends to infinity. Note that 
	\begin{equation}
	\label{eq:setproblemmonotone}
	\inf_{A\in\cF^{\alpha_n}}\sup_{\bP^*\in\cP}\bE^*[H\mathds{1}_{A}]\leq \inf_{A\in\cF^{\alpha_{n+1}}}\sup_{\bP^*\in\cP}\bE^*[H\mathds{1}_{A}]\leq \sup_{\bP^*\in\cP}\bE^*[H],
	\end{equation}
	because $\cF^{\alpha_{n+1}}\subset\cF^{\alpha_n}$.
	Therefore, the limit of $(\inf_{A\in\cF^{\alpha_n}}\sup_{\bP^*\in\cP}\bE^*[H\mathds{1}_{A}])_{n\in\bN}$ exists because the sequence is monotone and bounded. Let $\varepsilon>0$ be arbitrary. For each $n\in\bN$ exists $A_n\in\cF^{\alpha_n}$ such that 
	\begin{equation}
	\label{eq:setproblemapprox}
	\inf_{A\in\cF^{\alpha_n}}\sup_{\bP^*\in\cP}\bE^*[H\mathds{1}_A]\leq \sup_{\bP^*\in\cP}\bE^*[H\mathds{1}_{A_n}]<\inf_{A\in\cF^{\alpha_n}}\sup_{\bP^*\in\cP}\bE^*[H\mathds{1}_A]+\varepsilon.	
	\end{equation}	
%
Then, by Lemma\footnote{For random variables $(\xi_n)_{n\in\bN}\subset L^0(\Omega,\cF,\bP;\bR^d)$ we denote by $\text{conv}\{\xi_1,\xi_2,\dots\}$ the convex hull of $\xi_1,\xi_2,\dots$ which is defined $\omega$-wise. \begin{lemma*}[Lemma 1.70, \cite{follmerschied}]
\label{lem:lemma170}
	Let $(\xi_n)_{n\in\bN}$ be a sequence in $L^0(\Omega,\cF,\bP;\bR^d)$ such that $\sup_{n\in\bN}|\xi_n|<\infty$ $\bP$-a.s. Then there exists a sequence of convex combinations
	\[
	\eta_n\in\text{conv}\{\xi_n,\xi_{n+1},\dots\},\quad n\in\bN,
	\]
	which converges $\bP$-almost surely to some $\eta\in L^0(\Omega,\cF,\bP;\bR^d)$.
\end{lemma*}} 1.70 of \cite{follmerschied} there exists a sequence $\psi_n\in\text{conv}\{\mathds{1}_{A_n},\mathds{1}_{A_{n+1}},\dots\}$, $n\in\bN$, which converges $\bP$-a.s. to some $\psi\in L^\infty([\Omega,\cF_T,\bP;[0,1])$. Note that it is not clear if $\psi$ is an indicator function of some $\cF_T$-measurable set. We will show that $\psi=1$ $\bP$-a.s. For $n\in\bN$, $\psi_n$ is of the form 
	\begin{equation}
	\label{eq:psirepresentation}
		\psi_n=\sum_{k=n}^\infty \lambda_k^n \mathds{1}_{A_k},
	\end{equation}
	for some $(\lambda_k^n)_{k=n}^\infty\geq 0$ such that $\sum_{k=n}^\infty\lambda_k^n=1$. By dominated convergence and \eqref{eq:psirepresentation} we obtain
	\begin{equation}
	\label{eq:setproblempsi1}
	\bE_\bP[\psi]=\lim_{n\to\infty}\bE_\bP[\psi_n]=\lim_{n\to\infty}\bE_\bP\left[\sum_{k=n}^\infty\lambda_k^n\mathds{1}_{A_k}\right]=\lim_{n\to\infty}\left(\sum_{k=n}^\infty\lambda_k^n\bE_\bP\left[\mathds{1}_{A_k}\right]\right).
	\end{equation}
	Because $\sum_{k=n}^\infty\lambda_k^n=1$ and by the definition of the limes inferior, equation \eqref{eq:setproblempsi1} yields
	\begin{align}
	\nonumber
	\bE_\bP[\psi]&\geq \lim_{n\to\infty}\left(\sum_{k=n}^\infty\lambda_k^n\inf_{l\geq n}\bE_\bP\left[\mathds{1}_{A_l}\right]\right)
	=\lim_{n\to\infty}\left(\inf_{l\geq n}\bE_\bP\left[\mathds{1}_{A_l}\right]\right)\\
	\label{eq:setproblempsi2}
	&=\liminf_{n\to\infty}\bE_\bP[\mathds{1}_{A_n}]=\liminf_{n\to\infty}\bP(A_n)\geq \liminf_{n\to\infty}\alpha_n =1.
	\end{align}
Since $0\leq \psi\leq 1$, it follows that $\psi=1$ $\bP$-a.s. By \eqref{eq:setproblemapprox} and with similar arguments as in \eqref{eq:setproblempsi1} and \eqref{eq:setproblempsi2} using the supremum instead of the infimum, we obtain by dominated convergence for any $\bP^*\in\cP$ that
	\begin{align}
	\nonumber
		\limsup_{n\to\infty}\left(\inf_{A\in\cF^{\alpha_n}}\sup_{\bP^*\in\cP}\bE^*[H\mathds{1}_A]+\varepsilon\right)&\geq \limsup_{n\to\infty}\sup_{\bP^*\in\cP}\bE^*[H\mathds{1}_{A_n}]\\ 
		\label{eq:setproblemliminf}
		&\geq \lim_{n\to\infty}\bE^*[H\psi_n]=\bE^*[H\psi]=\bE^*[H].
	\end{align}
	Since the limit on the left hand side in \eqref{eq:setproblemliminf} exists by \eqref{eq:setproblemmonotone} and \eqref{eq:setproblemliminf} holds for all $\bP^*\in\cP$, we get
	\begin{equation}
	\label{eq:setproblemliminf2}
	\lim_{n\to\infty}\left(\inf_{A\in\cF^{\alpha_n}}\sup_{\bP^*\in\cP}\bE^*[H\mathds{1}_A]+\varepsilon\right)\geq \sup_{\bP^*\in\cP}\bE^*[H].
	\end{equation}
	Thus, we observe that \eqref{eq:setproblemmonotone} and \eqref{eq:setproblemliminf2} yields
	\[
	\lim_{n\to\infty}\left(\inf_{A\in\cF^{\alpha_n}}\sup_{\bP^*\in\cP}\bE^*[H\mathds{1}_A]\right)\leq \sup_{\bP^*\in\cP}\bE^*[H]\leq \lim_{n\to\infty}\left(\inf_{A\in\cF^{\alpha_n}}\sup_{\bP^*\in\cP}\bE^*[H\mathds{1}_A]+\varepsilon\right).
	\]
	As $\varepsilon>0$ was arbitrary this implies that
	\[
	\lim_{n\to\infty}\left(\inf_{A\in\cF^{\alpha_n}}\sup_{\bP^*\in\cP}\bE^*[H\mathds{1}_A]\right)=\sup_{\bP^*\in\cP}\bE^*[H].
	\]
\end{proof}\noindent

\subsubsection{Success ratios}
\label{sec:successratio}

Let $\cR:=L^\infty(\Omega,\cF_T,\bP;[0,1])$ be the set of randomized tests. For $\alpha\in(0,1)$ we denote by $\cR^\alpha$ the set
\[
\cR^\alpha:=\lbrace\varphi\in\cR:\ \bE_\bP[\varphi]\geq \alpha\rbrace.
\]
We now consider the following minimization problem 
\begin{equation}
	\label{eq:quantilehedgingtestfunctions}
	\inf\left\{ \sup_{\bP^*\in\cP}\bE^*\left[H\varphi\right]:\varphi\in\cR^\alpha\right\}.
\end{equation}
In a first step, we prove that this problem admits an explicit solution. In a second step, we show that the solution is given by the so-called success ratio, see Definition \ref{def:successratio} below. In particular, \eqref{eq:quantilehedgingtestfunctions} can be formulated in terms of success ratios, see also \cite{foellmer99}. In Proposition \ref{prop:testfunctionsolution} and \ref{prop:successratiosolution} we provide a proof for some result of \cite{foellmer99} for the sake of completeness.

\begin{proposition}
	\label{prop:testfunctionsolution}
	There exists a randomized test $\tilde\varphi\in\cR$ such that
	\[
	\bE_\bP[\tilde\varphi]=\alpha,
	\]
	and
	\begin{equation}
	\label{eq:minimizationthm}	
	\inf_{\varphi \in\cR^\alpha}\sup_{\bP^*\in\cP}\bE^*[H\varphi]=\sup_{\bP^*\in\cP}\bE^*[H\tilde \varphi].
	\end{equation}
\end{proposition}

\begin{proof}
	Take a sequence $(\varphi_n)_{n\in\bN}\subset \cR^\alpha$ such that
	\begin{equation}
	\label{eq:successratioapproxlim}
	\lim_{n\to\infty}\sup_{\bP^*\in\cP}\bE^*[H\varphi_n]=\inf_{\varphi\in\cR^\alpha}	\sup_{\bP^*\in\cP}\bE^*[H\varphi].
	\end{equation}
	By Lemma 1.70 of \cite{follmerschied} there exists a sequence of convex combinations $\tilde\varphi_n\in\text{conv}\{\varphi_n,\varphi_{n+1},\dots\}$ converging $\bP$-a.s. to a function $\tilde\varphi\in\cR$ because $\varphi_n\in [0,1]$ for all $n\in\bN$. Clearly $\tilde\varphi_n\in\cR^\alpha$ for each $n\in\bN$. Hence, dominated convergence yields that
	\begin{equation}
	\label{eq:tildeconstraint}
	\bE_\bP[\tilde\varphi]=\lim_{n\to\infty}\bE_\bP[	\tilde\varphi_n]\geq \alpha,
	\end{equation}
	and we get that $\tilde\varphi\in\cR^\alpha$. In the following we use similar arguments as in the proof of Theorem \ref{thm:convergencesuccessset}. In particular, $\tilde{\varphi}_n$ is of the form
	\begin{equation}
	\label{eq:testfunctionrepresentation}
	\tilde\varphi_n=\sum_{k=n}^\infty\lambda_k^n \varphi_k,	
	\end{equation}
	for some $(\lambda_k)_{k=n}^\infty$ such that $\sum_{k=n}^\infty\lambda_k^n=1$. By \eqref{eq:testfunctionrepresentation} we obtain for any $\bP^*\in\cP$ that
	\begin{equation}
	\label{eq:limsupconvcomb}
	\limsup_{n\to\infty} \bE^*\left[H\varphi_n\right]=\lim_{n\to\infty}\left(\sup_{k\geq n}\bE^*\left[H\varphi_k\right]\right)\geq \lim_{n\to\infty}\left(\sum_{k=n}^\infty\lambda_k^n\bE^*\left[H\varphi_k\right]\right)=\lim_{n\to\infty}\bE^*\left[H\tilde\varphi_n\right]= \bE^*\left[H\tilde\varphi\right],
	\end{equation}
	where we used monotone convergence. Moreover, we obtain by \eqref{eq:successratioapproxlim}, \eqref{eq:limsupconvcomb} and dominated convergence that
	\begin{equation}
	\label{eq:optimizerinquality}
		\inf_{\varphi\in\cR^\alpha}\sup_{\bP^*\in\cP}\bE^*[H\varphi]=\limsup_{n\to\infty}\sup_{\bP^*\in\cP}\bE^*[H\varphi_n]\geq \limsup_{n\to\infty}\bE^*[H\varphi_n]\geq \lim_{n\to\infty}\bE^*[H\tilde\varphi_n]=\bE^*[H\tilde\varphi].	
		\end{equation}
	Since \eqref{eq:optimizerinquality} holds for all $\bP^*\in\cP$ we obtain
	\[
	\inf_{\varphi\in\cR^\alpha}\sup_{\bP^*\in\cP}\bE^*[H\varphi]\geq \sup_{\bP^*\in\cP}\bE^*[H\tilde\varphi].
	\]
	Furthermore, $\tilde\varphi\in\cR^\alpha$ by \eqref{eq:tildeconstraint} yields
	\[
	\inf_{\varphi\in\cR^\alpha}\sup_{\bP^*\in\cP}\bE^*[H\varphi]= \sup_{\bP^*\in\cP}\bE^*[H\tilde\varphi].
	\]
	So $\tilde\varphi$ is the desired minimizer.\\
	We now show that $\bE_\bP[\tilde\varphi]=\alpha$ holds. If $\bE_\bP[\tilde\varphi]>\alpha$, then we can find $\varepsilon>0$ such that $\varphi_\varepsilon:=(1-\varepsilon)\tilde\varphi\in\cR^\alpha$, and 
	\begin{equation}
	\label{eq:varphitildealpha}
	\sup_{\bP^*\in\cP}\bE^*[H\varphi_\varepsilon]=(1-\varepsilon)\sup_{\bP^*\in\cP}\bE^*[H\tilde\varphi]<\sup_{\bP^*\in\cP}\bE^*[H\tilde\varphi],
	\end{equation}
	which contradicts the minimality property of $\tilde\varphi$. Thus, 
	\[
	\label{eq:varphitildealpha1}
	\bE_\bP[\tilde\varphi]=\alpha.
	\]
	\end{proof}\noindent 

\begin{definition}
\label{def:successratio}
 	For an admissible strategy with value process $V\in\cV$ we define its success ratio by
 	\begin{equation}
 	\label{eq:successratiodef}
 	\varphi_V:=\mathds{1}_{\{V_T\geq H\}}+\frac{V_T}{H}\mathds{1}_{\{V_T<H\}}.	
 	\end{equation}
 	For $\alpha\in (0,1)$ we denote by $\cV^\alpha$ the set 
 	\[
 	\cV^\alpha:=\left\{\varphi_V\in\cR:\ V\in\cV,\ \bE_\bP\left[\varphi_V\right]\geq \alpha\right\}.
 	\]
\end{definition}\noindent

\begin{remark}
	Note that for $V\in\cV$ we have that $V_T\geq 0$ $\bP$-a.s. In particular, $\bP(\{H=0\}\cap \{V_T<H\})=0$ and hence \eqref{eq:successratiodef} is well-defined. 
\end{remark}\noindent
In the following, we formulate the optimization problem \eqref{eq:quantilehedging} in terms of success ratios and prove that it is equivalent to \eqref{eq:quantilehedgingtestfunctions}, see Proposition \ref{prop:successratiosolution} below.\\
Consider the minimization problem
\begin{equation}
	\label{eq:quantilehedgingsuccessratios}
	\inf\left\{\sup_{\bP^*\in\cP}\bE^*\left[\varphi_V\right]:\ V\in\cV^\alpha\right\}.
\end{equation} 

\begin{proposition}
\label{prop:successratiosolution}
	There exists an admissible strategy with value process $\tilde V$ such that
	\[
	\bE_\bP\left[\varphi_{\tilde V}\right]=\alpha,
	\]
	and 
	\begin{equation}
	\label{eq:successratiominimizer}
		\inf_{V\in\cV^\alpha}\sup_{\bP^*\in\cP}\bE^*\left[H\varphi_{V}\right]=\sup_{\bP^*\in\cP}\bE^*\left[H\varphi_{\tilde V}\right],
	\end{equation}
	where $\varphi_V$ denotes the success ratio associated to a portfolio $V\in\cV$ as in \eqref{eq:successratiodef}. Moreover, $\varphi_{\tilde V}$ coincides with the solution $\tilde \varphi$ from Proposition \ref{prop:testfunctionsolution}.
\end{proposition}

\begin{proof}
	Note that 
	\[
	\left\{\varphi_V\in\cR:\ V\in\cV^\alpha\right\}\subseteq \cR^\alpha,
	\]
	and thus 
	\begin{equation}
	\label{eq:successratiossolutionviatest}
	\inf_{\varphi\in\cR^\alpha}\sup_{\bP^*\in\cP}\bE^*[H\varphi]\leq \inf_{V\in\cV^\alpha}\sup_{\bP^*\in\cP}\bE^*\left[H\varphi_{V}\right].
	\end{equation}
	By Proposition \ref{prop:testfunctionsolution}, we know that the left hand side of \eqref{eq:successratiossolutionviatest} admits a solution $\tilde\varphi\in\cR$. We prove that there exists $\tilde V\in\cV^\alpha$ such that
	\[
	\tilde\varphi=\varphi_{\tilde V}\quad \bP\text{-a.s.}
	\]
	Define the the modified claim
	\[
	\tilde H:=H\tilde\varphi.
	\]
	By Theorem 7.13 of \cite{follmerschied} there exists a minimal superhedging strategy $\tilde\xi$ with value process $\tilde V$ for $\tilde H$ such that 
	\[
	\tilde V_0 = \sup_{\bP^*\in\cP}\bE^*\left[\tilde H\right].
	\]
	First, $\tilde\xi$ can be assumed to be admissible by Remark \ref{rem:admissiblestrategies} and hence $\tilde V\in\cV$. Now, we show that $\tilde V\in\cV^\alpha$. We have
\begin{equation}
\label{eq:varphitildeinequality1}
\varphi_{\tilde V}=\mathds{1}_{\{\tilde V_T\geq H\}}+\frac{\tilde V_T}{H}\mathds{1}_{\{\tilde V_T<H\}}\geq \tilde\varphi\mathds{1}_{\{\tilde V_T\geq H\}} + \frac{H\tilde\varphi}{H}\mathds{1}_{\{\tilde V_T< H\}}=\tilde\varphi,
\end{equation}
where we used that $\tilde V$ is the value process of the minimal superhedging strategy of $\tilde H=H\tilde\varphi$ and $0\leq \tilde\varphi\leq 1$. Therefore, we get
\[
\bE_\bP[\varphi_{\tilde{V}}]\geq \bE_\bP	[\tilde\varphi]\geq \alpha,
\]
so $\tilde V\in\cV^\alpha$ and $\varphi_{\tilde{V}}\in\cR^\alpha$. It is left to show that $\tilde\varphi=\varphi_{\tilde V}$ $\bP$-a.s. By \eqref{eq:varphitildeinequality1} we obtain $\varphi_{\tilde V}\geq \tilde\varphi$. For the reverse direction we first show that $\varphi_{\tilde V}$ is also a minimizer of the problem \eqref{eq:minimizationthm}, i.e.,
\[
\sup_{\bP^*\in\cP}\bE^*[H\varphi_{\tilde V}]\leq \sup_{\bP^*\in\cP}\bE^*[H\tilde\varphi].	
\]
Indeed, since $\tilde V$ is the value process of an admissible strategy, $V$ is a $\bP^*$-martingale for all $\bP^*\in\cP$ by Theorem 5.14 of \cite{follmerschied} and thus we get that			
\begin{equation}
\label{eq:varphivminimal}
\sup_{\bP^*\in\cP}\bE^*[H\varphi_{\tilde V}]=\sup_{\bP^*\in\cP}\bE^*\left[H\left(\mathds{1}_{\{\tilde V_T\geq H\}}+\frac{\tilde V_T}{H}\mathds{1}_{\{\tilde V_T<H\}}\right)\right]\leq \sup_{\bP^*\in\cP}\bE^*[\tilde V_T]= \tilde V_0 = \sup_{\bP^*\in\cP}\bE^*[H\tilde{\varphi}],
\end{equation}
where we used in the last equality that $\tilde V_0$ is the superhedging price of $\tilde H=H\tilde\varphi$. In particular, $\varphi_{\tilde{V}}\in\cR^\alpha$ is a minimizer. By the same arguments as in \eqref{eq:varphitildealpha} it follows that
\begin{equation}
\label{eq:varphivalpha}
\bE_\bP[\varphi_{\tilde{V}}]=\alpha.
\end{equation}
Thus, we get by \eqref{eq:varphitildealpha} and \eqref{eq:varphivalpha} that
\[
\bE_\bP[\varphi_{\tilde V}]=\alpha=\bE_\bP[\tilde\varphi],	
\]
i.e., $\bE[\varphi_{\tilde V}-\tilde\varphi]=0$. Together with \eqref{eq:varphitildeinequality1}, this implies $\varphi_{\tilde V}=\tilde\varphi$ $\bP$-a.s. We have proved that $\tilde V\in\cV^\alpha$ and
\[
\sup_{\bP^*\in\cP}\bE^*\left[H\varphi_{\tilde V}\right]=\inf_{\varphi\in\cR^\alpha}\sup_{\bP^*\in\cP}\bE^*[H\varphi]\leq \inf_{V\in\cV^\alpha}\sup_{\bP^*\in\cP}\bE^*\left[H\varphi_{V}\right].
\]
In particular, $\varphi_{\tilde V}$ solves \eqref{eq:successratiominimizer} and the quantile hedging formulations of \eqref{eq:quantilehedgingtestfunctions} and \eqref{eq:quantilehedgingsuccessratios} are equivalent.
\end{proof}\noindent

\begin{corollary}
\label{cor:successratioconvergence}
	The following convergence holds:
	\[
	\inf_{V\in\cV^\alpha}\sup_{\bP^*\in\cP}\bE^*[H\varphi_V]\xrightarrow{\alpha\uparrow 1} \sup_{\bP^*\in\cP}\bE^*[H],	
	\]
	where $\varphi_V$ denotes the success ratio associated to a portfolio $V\in\cV$ as in \eqref{eq:successratiodef}.
\end{corollary}

\begin{proof}
The proof is similar to the one of Theorem \ref{thm:convergencesuccessset} and is omitted.
\end{proof}\noindent

\subsection{Neural network approximation for $t=0$}
\label{sec:NNprice0}
We now study how to approximate the superhedging price at $t=0$ by using neural networks.  \\
We recall the following definition, see e.g.\ \cite{buehler2018}. Common choices for $\sigma$ below are $\sigma(x)=\frac{1}{1-e^{-x}}$ and $\sigma(x)=\tanh(x)$.
\begin{definition}\label{def:nn} Consider $L, N_0,N_1,\ldots,N_L \in \N$ with $L \geq 2$, $\activF \colon (\R,\cB(\bR)) \to (\R,\cB(\bR))$ measurable and for any $\ell=1,\ldots,L$, let
	$W_\ell \colon \R^{N_{\ell-1}} \to \R^{N_\ell}$ be an affine
	function. A function $F \colon \R^{N_0} \to \R^{N_L}$ defined as
	\[ F(x)=W_L \circ F_{L-1} \circ \cdots \circ F_1 \text{ with }
	F_\ell = \activF \circ W_\ell \, \text{ for } \ell=1,\ldots,L-1, \]
	is called a \emph{(feed forward) neural network}. Here the
	\textit{activation function} $\activF$ is applied componentwise. $L$
	denotes the number of layers, $N_1,\ldots,N_{L-1}$ denote the
	dimensions of the hidden layers and $N_0$, $N_L$ the dimension of the input and output layers, respectively. For any $\ell=1,\ldots,L$ the affine
	function $W_\ell$ is given as $ W_\ell(x) = A^\ell x + b^\ell$ for
	some $A^\ell \in \R^{N_\ell \times N_{\ell-1}}$ and
	$b^\ell \in \R^{N_\ell}$. For any
	$i=1,\ldots N_\ell, j=1,\ldots,N_{\ell-1}$ the number $A^\ell_{i j}$
	is interpreted as the weight of the edge connecting the node $i$ of
	layer $\ell-1$ to node $j$ of layer $\ell$. The number of non-zero
	weights of a network is $
	\sum_{\ell=1}^L \|A^\ell\|_0 + \|b^\ell\|_0$, i.e.\
	the sum of the number of non-zero entries of
	the matrices $A^\ell$, $\ell=1,\ldots,L$, and vectors $b^\ell$,
	$\ell=1,\ldots,L$.
\end{definition}\noindent
For $k=1,\ldots,T+1$ we denote the set of all possible neural network parameters corresponding to neural networks mapping $\R^{mk} \to \R^d$ by
\[
\Theta_k = \cup_{L \geq 2} \cup_{(N_0,\ldots,N_L) \in \{mk \} \times \N^{L-1} \times \{d\}} \left(\times_{\ell = 1}^{L} \R^{N_\ell \times N_{\ell-1}} \times \R^{N_\ell} \right).
\]
With $F^{\theta_k}$ we denote the neural network with parameters specified by $\theta_k\in\Theta_k$, see Definition \ref{def:nn}. Recall that $\cF_t=\sigma(Y_0,\dots,Y_t)=\sigma(\cY_t)$ for $t=0,\dots,T$, and for some $\bR^m$-valued stochastic process $Y$. Then, any $\cF_t$-measurable random variable $Z$ can be written as $Z=f_t(\cY_t)$ for some measurable function $f_t$. Using Theorem \ref{thm:universalapproxprob}, $f_t$ can be approximated by a deep neural network in a suitable metric.\\
The approximate superhedging price is then 
\begin{equation}\label{eq:NNprice}
\resizebox{15cm}{!}{
$\inf\cU_0^\Theta= \inf  \left\lbrace u\in \R \,:\, \exists \, \theta_{k,\xi} \in \Theta_k, k=1,\dots,T, \text{ s.t. }  u + \sum_{k=1}^T F^{\theta_{k,\xi}}(\Yc_{k-1}) \cdot (X_k - X_{k-1}) \geq H\ \bP\text{-a.s.} \right\rbrace.$
}
\end{equation}
For $\alpha\in (0,1)$ the approximate $\alpha$-quantile hedging price is then 
\begin{equation}\label{eq:NNpriceAlpha}
\resizebox{15cm}{!}{
$\inf\cU_0^{\Theta,\alpha}= \inf  \left\lbrace u \in \R \,:\, \exists \, \theta_{k,\xi} \in \Theta_k, k=1,\ldots,T \text{ s.t. } \P\left( u + \sum_{k=1}^T F^{\theta_{k,\xi}}(\Yc_{k-1}) \cdot (X_k - X_{k-1}) \geq H \right) \geq \alpha \right\rbrace.$}
\end{equation}
For $C>0$ we also define the truncated approximate superhedging price $\inf\cU_0^{\Theta,C}$ and the truncated approximate $\alpha$-quantile hedging price $\inf\cU_0^{\Theta,C,\alpha}$ with
\begin{equation}
\resizebox{15cm}{!}
{
$\cU_0^{\Theta,C}:=\left\lbrace u\in\bR:\exists\theta_{k,\xi}\in\Theta_k, k=1,\dots,T\text{ s.t. }u+\sum_{k=1}^T \left(\left(F^{\theta_{k,\xi}}\wedge C\right)\vee (-C)\right)(\cY_{t-1})\cdot(X_k-X_{k-1})\geq H\, \bP\text{-a.s.}\right\rbrace$}
\end{equation}
and
\begin{equation}
\resizebox{15cm}{!}{
$\cU_0^{\Theta,C,\alpha}:=\left\lbrace u\in\bR:\exists\theta_{k,\xi}\in\Theta_k, k=1,\dots,T\text{ s.t. }\bP\left(u+\sum_{k=1}^T \left(\left(F^{\theta_{k,\xi}}\wedge C\right)\vee (-C)\right)(\cY_{t-1})\cdot(X_k-X_{k-1})\geq H\right)\geq\alpha\right\rbrace$},
\end{equation}
where the maximum and minimum are taken componentwise.

\begin{assumption}
	\label{ass:bdd}
	Suppose that 
	\[
	\resizebox{15cm}{!}
	{
	$\inf\cU_0=\inf\cU_0^{\text{bdd}}:=\inf\left\lbrace u\in\bR:\exists\xi\text{ pred. s.t. }\xi_k\in L^\infty\ \forall k\in\{1,\dots,T\},\ u+\sum_{k=1}^T\xi_k\cdot(X_k-X_{k-1})\geq H\ \bP\text{-a.s.}\right\rbrace,$}
	\]
	where $\|\cdot\|_\infty$ denotes the $L^\infty$-norm.
\end{assumption}
The next result shows that $\inf\cU_0^{\Theta,C,\alpha}$ can be used as an approximation of the superhedging price $\inf\cU_0$.

\begin{theorem}
\label{thm:universalapprox2} Assume $\sigma$ is bounded and non-constant. Further, suppose Assumption \ref{ass:bdd} is fulfilled. Then for any $\varepsilon>0$ there exists $\alpha=\alpha(\varepsilon)\in (0,1)$ and $C=C(\varepsilon)\in (0,\infty)$ such that
	\begin{equation}
	\label{eq:NNpriceAlphaApproxB}
	\inf\cU_0+\varepsilon\geq \inf\cU_0^{\Theta,C,\alpha}\geq \inf\cU_0-\varepsilon.	
	\end{equation}
\end{theorem}

\begin{proof} 
By Assumption \ref{ass:bdd} we can consider $\inf\cU_0^{\text{bdd}}$ instead of $\inf\cU_0$. Set $\tilde u_0=\inf\cU_0^{\text{bdd}}$. Then for $\varepsilon>0$ there exists a predictable strategy $ \tilde{\xi}$ such that $\sup_{1\leq k\leq T}\| \tilde\xi_k\|_\infty<\infty$ and $\tilde u_0 + \frac{\varepsilon}{2} +  \sum_{k=1}^T \tilde{\xi}_k \cdot (X_k - X_{k-1}) \geq H,\ \P\text{-a.s.}$ Define $C=C(\varepsilon)$ by
\begin{equation}
\label{eq:NNdefC}
C:=\sup_{1\leq k\leq T}\| \tilde\xi_k\|_\infty+1.
\end{equation}
Further, for $\alpha\in (0,1]$ define $\cU_0^{C,\alpha}$ by
\[
\cU_0^{C,\alpha}:=\left\lbrace u\in\bR:\exists\xi\text{ pred. s.t. }\sup_{1\leq k\leq T}\|\xi_k\|_\infty\leq C,\ \bP\left(u+\sum_{k=1}^T\xi_k\cdot(X_k-X_{k-1})\geq H\right)\geq \alpha\right\rbrace.
\]
First, we prove that the limit of $\inf\cU_0^{C,\alpha}$ for $\alpha$ tending to $1$ exists and that
\begin{equation}
\label{eq:NNobjectiveconv}
\inf\cU_0^{\text{bdd}}\leq \lim_{\alpha\to 1}\inf\cU_0^{C,\alpha}\leq \inf\cU_0^{\text{bdd}}+\varepsilon.
\end{equation}
Let $(\alpha_n)_{n\in\bN}\subset (0,1)$ be a sequence such that $\alpha_n\uparrow 1$ as $n$ tends to infinity. Then, $\inf\cU_0^{C,\alpha_n}\leq \inf\cU_0^{C,\alpha_{n+1}}\leq \inf\cU_0^{C,1}=:\inf\cU_0^C$ since
\[
\cU_0^{C,\alpha_n}\supset\cU_0^{C,\alpha_{n+1}},
\]
and therefore $u_n\leq u_{n+1}$, where $u_n:=\inf\cU_0^{C,\alpha_n}$, for $n\in\bN$. Thus, the limit $u^C=\lim_{n\to\infty} u_n$ is well-defined and $u^C\leq \inf\cU_0^C$. Furthermore, for $n\in\bN$ and $\delta>0$, there exists $\xi^{(n)}$ predictable such that $\sup_{1\leq k\leq T}\|\xi_k^{(n)}\|_\infty \leq C$ and
\begin{equation}
\label{eq:increasingprobability}
\bP\left(u_n+\delta+\sum_{k=1}^T\xi_k^{(n)}\cdot(X_k-X_{k-1})\geq H\right)\geq \alpha_n.
\end{equation}
For $n\in\bN$, define $A_n\in\cF_T$ by
\[
A_n:=\left\lbrace u_n+\delta+\sum_{k=1}^T\xi_k^{(n)}\cdot(X_k-X_{k-1})\geq H\right\rbrace.
\]
Then $\bP(A_n)\geq \alpha_n$ and hence $\bP(A_n)\uparrow 1$ as $n$ tends to infinity. Since $\sup_{1\leq k\leq T}\|\xi_k^{(n)}\|_\infty\leq C$ for all $n\in\bN$ we get by Theorem 5.14 of \cite{follmerschied} for any $\bP^*\in\cP$ that
\begin{align}
\label{eq:proof0conv1}
	u_{n}+\delta &=\bE^*\left[u_n+\delta+\sum_{k=1}^T\xi_k^{(n)}\cdot(X_k-X_{k-1})\right]\\
	\nonumber
	&\geq \bE^*\left[H\mathds{1}_{A_n}\right]+\bE^*\left[\left(u_n+\delta+\sum_{k=1}^T\xi_k^{(n)}\cdot(X_k-X_{k-1})\right)\mathds{1}_{A_n^c}\right]\\
	\nonumber
&\geq \bE^*[H\mathds{1}_{A_n}]+\bE^*\left[\left(u_n+\delta-\sum_{k=1}^T \sum_{i=1}^d |\xi_k^{i,(n)}| |X_k^i-X_{k-1}^i|\right)\mathds{1}_{A_n^c}\right]\\
\label{eq:proof0conv2}
&\geq \bE^*[H\mathds{1}_{A_n}]+\bE^*\left[\left(u_n+\delta-C\sum_{k=1}^T \sum_{i=1}^d|X_k^i-X_{k-1}^i|\right)\mathds{1}_{A_n^c}\right].
\end{align}
Recall that $X=(X^1,\dots,X^d)$ is a $d$-dimensional $\bP^*$-martingale, $u_n\leq u^C\leq \tilde u_0+\frac{\varepsilon}{2}$, $n\in\bN$ and thus for all $n\in\bN$
\[
\left|u_n+\delta-C\sum_{k=1}^T\sum_{i=1}^d |X_k^i-X_{k-1}^i|\right|\leq \left(\left|u_C+\delta\right|+|C|\sum_{k=1}^T\sum_{i=1}^d \|X_k-X_{k-1}\|\right)\in L^1(\Omega,\cF_T,\bP^*).
\]
Furthermore, $\mathds{1}_{A_n}$ converges to  $1$ in probability as $n$ tends to infinity, since for any $\gamma \in (0,1)$ we have
\[
\bP\left(\left|\mathds{1}_{A_n}-1\right|>\gamma\right)=\bP\left(\mathds{1}_{A_n^c}>\gamma\right)=\bP(A_n^c)\xrightarrow{n\to\infty} 0,
\]
because of \eqref{eq:increasingprobability}. By dominated convergence we obtain that
\[
\lim_{n\to\infty}\bE^*[H\mathds{1}_{A_{n}}]=\bE^*[H],
\]
and 
\[
\lim_{n\to\infty}\bE^*\left[\left(u_{n}+\delta-C\sum_{k=1}^T\sum_{i=1}^d|X_k^i-X_{k-1}^i|\right)\mathds{1}_{A_{n}^c}\right]=0.
\]
Note that for dominated convergence, it is sufficient that $\mathds{1}_{A_n}$ converges only in probability.
Taking $n$ to infinity in \eqref{eq:proof0conv1} and \eqref{eq:proof0conv2} yields
\begin{equation}
\label{eq:limitC}
\resizebox{15cm}{!}{
$\lim_{n\to\infty}u_{n}+\delta=u^C+\delta \geq \lim_{n\to\infty}\left(\bE^*[H\mathds{1}_{A_{n}}]+\bE^*\left[\left(u_{n}+\delta-C\sum_{k=1}^T\sum_{i=1}^d|X_k^i-X_{k-1}^i|\right)\mathds{1}_{A_{n}^c}\right]\right)=\bE^*[H].$}
\end{equation}
As \eqref{eq:limitC} holds for all $\bP^*\in\cP$ we get by the superhedging duality that
\[
\lim_{n\to\infty}\inf\cU_0^{C,\alpha_n}+\delta=u^C+\delta \geq \sup_{\bP^*\in\cP}\bE^*[H]=\inf\cU_0=\inf\cU_0^{\text{bdd}}.
\]
Because $\delta>0$ was arbitrary, this implies 
\[
\lim_{\alpha\to 1}\inf\cU_0^{C,\alpha}\geq \inf\cU_0=\inf\cU_0^{\text{bdd}}.
\]
To conclude the proof of \eqref{eq:NNobjectiveconv}, we note that $(\tilde u_0+\frac{\varepsilon}{2})\in\cU_0^C$ by definition and that $\cU_0^C\subset \cU_0^{\text{bdd}}$. This implies that $\inf \cU_0^{\text{bdd}}\leq \inf\cU_0^C$, and
\[
\lim_{\alpha\to 1}\inf\cU_0^{C,\alpha}\leq \inf\cU_0^C\leq \tilde u_0+\frac{\varepsilon}{2}\leq  \inf\cU_0^{\text{bdd}}+\varepsilon,
\]
hence \eqref{eq:NNobjectiveconv} follows. We observe that $\cU_0^{\Theta,C,\alpha}\subset \cU_0^{C,\alpha}$ for all $\alpha\in (0,1)$. Furthermore, by \eqref{eq:NNobjectiveconv} for $\varepsilon>0$ there exists $\alpha=\alpha(\varepsilon)\in (0,1)$ such that 
\begin{equation}
\label{eq:NNalpha}
\inf\cU_0-\varepsilon=\inf\cU_0^{\text{bdd}}-\varepsilon\leq \inf\cU_0^{C,\alpha}\leq \inf\cU_0^{\Theta,C,\alpha},
\end{equation}
which proves the second inequality in \eqref{eq:NNpriceAlphaApproxB}.\\
To prove the first inequality in \eqref{eq:NNpriceAlphaApproxB}, let $\alpha$ be given. Consider
\[
M_n = \left\lbrace \tilde u_0 + \frac{\varepsilon}{2} +  \sum_{k=1}^T \tilde{\xi}_k \cdot (X_k - X_{k-1}) \geq H  \right\rbrace \cap \{ \|X_k-X_{k-1}\| \leq n \text{ for } k=1,\ldots,T\},
\]
for $n\in\bN$. Then $M_n \subset M_{n+1}$ and therefore by continuity from below
\[
1= \P\left(\tilde u_0 + \frac{\varepsilon}{2} +  \sum_{k=1}^T \tilde{\xi}_k \cdot (X_k - X_{k-1}) \geq H \right) = \P(\cup_{n \in \N} M_n ) = \lim_{n \to \infty} \P(M_n).
\]
Thus, we may choose $n \in \N$ such that $\P(M_n) \geq \frac{\alpha+1}{2}$.
As $ \tilde{\xi}$ is predictable, for each $k=1,\ldots,T$ there exists a measurable function $f_k \colon (\R^{mk},\cB(\bR^{mk})) \to (\R^d,\cB(\bR^d))$ such that $\tilde{\xi}_k = f_k(\Yc_{k-1})$. By the universal approximation theorem \cite[Theorem~1 and Section~3]{hornik1991}, see also Theorem \ref{thm:universalapproxprob} in the appendix, with measure $\mu$ given by the law of $\Yc_{k-1}$ under $\P$, for each $k=1,\ldots,T$ there exists $\theta_{k,\tilde\xi} \in \Theta$ such that 
\begin{equation}
\label{eq:NNapproxD}
\P(D_k) < \frac{1-\alpha}{2T}, \quad \text{ where } D_k =\left\lbrace \omega \in \Omega \colon \|f_k(\Yc_{k-1}(\omega))-F^{\theta_{k,\tilde\xi}}(\Yc_{k-1}(\omega))\|>\left(\frac{\varepsilon}{2nT}\wedge\frac{1}{2}\right)\right\rbrace.
\end{equation}
Define 
\[
\tilde F^{\theta_{k,\tilde\xi}}:=\left(F^{\theta_{k,\tilde\xi}}\wedge C\right)\vee (-C),\quad k=1,\dots,T.
\]
By the definition of $C$ in \eqref{eq:NNdefC}, we get that
\[
\|\tilde\xi_k\|_\infty+\left(\frac{\varepsilon}{2nT}\wedge\frac{1}{2}\right)< C \quad \text{ for all }k=1,\dots,T.
\]
On $D_k^c$ we have for $i\in \{1,\dots,d\}$ that
\[
\left|F_i^{\theta_{k,\tilde\xi}}(\cY_{k-1})\right|\leq \left\|F^{\theta_{k,\tilde\xi}}(\cY_{k-1})\right\|\leq \left\|\tilde\xi_k\right\|_\infty+\left(\frac{\varepsilon}{2nT}\wedge \frac{1}{2}\right)<C, 
\] 
and hence $\tilde F_i^{\theta_{k,\tilde\xi}}(\cY_{k-1})=F_i^{\theta_{k,\tilde\xi}}(\cY_{k-1})$ on $D_k^c$. Conversely, for $\omega\in\Omega$ such that
\[
\|f_k(\Yc_{k-1}(\omega))-\tilde F^{\theta_{k,\tilde\xi}}(\Yc_{k-1}(\omega))\|\leq\left(\frac{\varepsilon}{2nT}\wedge\frac{1}{2}\right),
\]
we get for $i\in\{1,\dots,d\}$ that
\[
\left|\tilde F_i^{\theta_{k,\tilde\xi}}(\cY_{k-1}(\omega))\right|\leq \left\|\tilde F^{\theta_{k,\tilde\xi}}(\cY_{k-1}(\omega))\right\|\leq \left\|\tilde\xi_k\right\|_\infty+\left(\frac{\varepsilon}{2nT}\wedge \frac{1}{2}\right)<C,
\]
and hence $\tilde F_i^{\theta_{k,\tilde\xi}}(\cY_{k-1}(\omega))=F_i^{\theta_{k,\tilde\xi}}(\cY_{k-1}(\omega))$. In particular,
\[
\resizebox{15cm}{!}
{
$\left\lbrace \omega\in\Omega\colon \|f_k(\Yc_{k-1}(\omega))-\tilde F^{\theta_{k,\tilde\xi}}(\Yc_{k-1}(\omega))\|\leq\left(\frac{\varepsilon}{2nT}\wedge\frac{1}{2}\right)\right\rbrace=\underbrace{\left\lbrace \omega\in\Omega\colon \|f_k(\Yc_{k-1}(\omega))-F^{\theta_{k,\tilde\xi}}(\Yc_{k-1}(\omega))\|\leq\left(\frac{\varepsilon}{2nT}\wedge\frac{1}{2}\right)\right\rbrace}_{=D_k^c},$}
\]
for all $k=1,\dots,T$. Therefore, we get that $D_k=\tilde D_k$ with
\[
\tilde D_k =\left\lbrace \omega \in \Omega \colon \|f_k(\Yc_{k-1}(\omega))-\tilde F^{\theta_{k,\tilde\xi}}(\Yc_{k-1}(\omega))\|>\left(\frac{\varepsilon}{2nT}\wedge\frac{1}{2}\right)\right\rbrace,
\]
and
\[
\P(\tilde D_k) < \frac{1-\alpha}{2T}.
\]
On $M_n \cap \tilde D_1^c \cap \ldots \cap \tilde D_T^c$ we have 
\[\begin{aligned}
\sum_{k=1}^T \tilde{\xi}_k \cdot (X_k - X_{k-1}) & = \sum_{k=1}^T (\tilde{\xi}_k-\tilde F^{\theta_{k,\tilde\xi}}(\Yc_{k-1})) \cdot (X_k - X_{k-1}) + \sum_{k=1}^T \tilde F^{\theta_{k,\tilde\xi}}(\Yc_{k-1}) \cdot (X_k - X_{k-1}) 
\\ & \leq \sum_{k=1}^T \|f_k(\Yc_{k-1})-\tilde F^{\theta_{k,\tilde\xi}}(\Yc_{k-1})\| \|X_k - X_{k-1}\| \\
&+ \sum_{k=1}^T \tilde F^{\theta_{k,\tilde\xi}}(\Yc_{k-1}) \cdot (X_k - X_{k-1})
\\ & \leq \frac{\varepsilon}{2} + \sum_{k=1}^T \tilde F^{\theta_{k,\tilde\xi}}(\Yc_{k-1}) \cdot (X_k - X_{k-1})
\end{aligned} \]
and therefore
\[
M_n \cap \tilde D_1^c \cap \ldots \cap \tilde D_T^c \subset \left\lbrace \tilde u_0 + \varepsilon +  \sum_{k=1}^T \tilde F^{\theta_{k,\tilde\xi}}(\Yc_{k-1}) \cdot (X_k - X_{k-1}) \geq H  \right\rbrace.
\]
This inclusion and the Fr\'echet inequalities\footnote{For $C_1,\dots,C_l\in\cF$ it holds that $P(C_1\cap\dots\cap C_l)\geq \max\{0,\bP(C_1)+\dots+\bP(C_l)-(l-1)\}$.} yield
\[\begin{aligned}
\P & \left( \tilde u_0 + \varepsilon +  \sum_{k=1}^T \tilde F^{\theta_{k,\tilde\xi}}(\Yc_{k-1}) \cdot (X_k - X_{k-1}) \geq H\right) 
\\ \geq & \P(M_n \cap \tilde D_1^c \cap \ldots \cap \tilde D_T^c)\\
 \geq &\P(M_n) + \P(\tilde D_1^c)+ \cdots + \P(\tilde D_T^c) - T \geq \frac{\alpha+1}{2} + T\left(1-\frac{1-\alpha}{2T}\right) - T = \alpha.
\end{aligned}\]
This proves the left inequality of \eqref{eq:NNpriceAlphaApproxB}.
\end{proof}

\begin{remark}
	Note that in the proof of Theorem \ref{thm:universalapprox2} we compute both the price at $t=0$ and the superhedging strategy for the complete interval. 
\end{remark}

\begin{remark}
	Thanks to the universal approximation theorem in \cite{hornik1991}, we could in fact restrict our attention to neural networks with one hidden layer and the result in Theorem~\ref{thm:universalapprox2} remains valid. Thus, for each $k=1,\ldots,T$ we could fix $L=2$, $N_0 = mk$, $N_2=d$ and consider instead the simpler parameter sets 
	\[\begin{aligned}
	\Theta_k & = \cup_{N_1 \in \N} (\R^{N_1 \times mk} \times \R^{N_1}) \times (\R^{d \times N_1} \times \R^d)
	\\ \Theta^C_k & =  ([-C,C]^{C \times mk} \times [-C,C]^{C}) \times ([-C,C]^{d \times N_1} \times [-C,C]^d).
	\end{aligned}
	\]
	Note the simpler form of $\Theta^C_k$, which is due to the fact that all one-hidden layer networks with $N_1 \leq C$ hidden nodes can be written as one-hidden layer networks with $C$ hidden nodes and appropriate weights set to $0$.
\end{remark}

\section{Superhedging price for $t>0$}
\label{sec:superhedgingt}
In this section we establish a method to approximate superhedging prices for $t>0$. Using a version of the uniform Doob decomposition, see Theorem 7.5 of \cite{follmerschied}, the problem reduces to the approximation of the so-called process of consumption. In the first part, we build the theoretical basis for this approach. In the second part we prove that this method can be used to approximate the superhedging price for $t>0$ by neural networks.

\subsection{Uniform Doob Decomposition}
\label{sec:uniformDoob}
We briefly summarize some results on superhedging in discrete time in Corollary \ref{cor:resultsfollmerschied} below. For a more detailed overview we refer to Chapter 7 of \cite{follmerschied}.\\
Recall that $H$ denotes a discounted European claim satisfying
\[
\sup_{\bP^*\in\cP}\bE^*[H]<\infty.
\]
The superhedging price at $t=0$, $\sup_{\bP^*\in\cP}\bE^*[H]$ and the associated strategy $\xi$ can be calculated as in Section \ref{sec:price0} and so we consider them as known. The remaining unknown component is the process of consumption $B$ given by \eqref{eq:uniformdoob}. By Corollary \ref{cor:resultsfollmerschied}, 
\[
\left(\esssup_{\bP^*\in\cP}\bE^*[H\mid\cF_t]\right)_{t=0,1,\dots,T}
\]
is the smallest $\cP$-supermartingale whose terminal value dominates $H$. Consider the stochastic process $\tilde B=(\tilde B_t)_{t=0,\dots,T}$ defined as $\tilde B_0:=0$ and for $t=1,\dots,T$,
\begin{equation}
	\label{eq:optimizationB}
	\tilde B_t:=\esssup\cB_t,
\end{equation}
where
\begin{equation}
\label{eq:optimizationBset}
\cB_t:=\left\{D_t\in L^0(\Omega,\cF_t,\bP):\tilde B_{t-1}\leq D_t\leq \sup_{\bP^*\in\cP}\bE^*[H]+\sum_{k=1}^T\xi_k\cdot(X_k-X_{k-1})-H\ \bP\text{-a.s.}\right\}.
\end{equation}

\begin{proposition}
	\label{prop:optimizationB}
	We have that
	\[
	B_t=\tilde B_t\quad\bP\text{-a.s., for all }t=0,\dots,T,
	\]
	where $B$ is given in \eqref{eq:uniformdoob} and $\tilde B$ in \eqref{eq:optimizationBset}, respectively.
\end{proposition}

\begin{proof}
	The proof follows by induction. For $t=0$ we have $B_0=0=\tilde B_0$ by definition. For the induction step assume that $B_{t-1}=\tilde B_{t-1}$ $\bP$-a.s. for some $1\leq t\leq T$. First we observe that $B_t\geq \tilde B_{t-1}$ because $B$ is increasing and by the assumption of the induction step. In addition, by \eqref{eq:optimization} we obtain 
	\begin{equation}
	\label{eq:proofBdom}
	\sup_{\bP^*\in\cP}\bE^*[H]+\sum_{k=1}^T\xi_k\cdot(X_k-X_{k-1})-H\geq B_t.
	\end{equation}
	In particular, $B_t\in \cB_t$ and thus $B_t\leq \tilde B_t$ $\bP$-a.s. Assume that $\bP(B_{t}<\tilde B_{t})>0$. Then define $\tilde V=(\tilde V_s)_{s=0,\dots,T}$ by
	\[
	\tilde V_s:=\sup_{\bP^*\in\cP}\bE^*[H]+\sum_{k=1}^s\xi_k\cdot (X_k-X_{k-1})-\tilde B_s
	\]
	First, we note that
	\[
	\sup_{\bP^*\in\cP}\bE^*[H]+\sum_{k=1}^T\xi_k\cdot (X_k-X_{k-1})\geq H\geq 0\quad\bP\text{-a.s.}
	\]
	and thus by Theorem 5.14 of \cite{follmerschied} we have for any $\bP^*\in\cP$ that 
	\[\left(\sup_{\bP^*\in\cP}\bE^*[H]+\sum_{k=1}^s\xi_k\cdot (X_k-X_{k-1})\right)_{s=0,\dots,T}
	\]
	is $\bP^*$-martingale for all $\bP^*\in\cP$. Further, by \eqref{eq:proofBdom} and \eqref{eq:uniformdoob} we obtain
	\[
	0\leq \tilde B_s\leq \sup_{\bP^*\in\cP}\bE^*[H]+\sum_{k=1}^T\xi_k\cdot (X_k-X_{k-1})-H\quad \text{for all }T=0,\dots,T,\ \bP^*\in\cP,
	\]
	and 
	\[
	\sup_{\bP^*\in\cP}\bE^*[H]+\sum_{k=1}^T\xi_k\cdot (X_k-X_{k-1})-H\in L^1(\Omega,\cF,\bP^*)\quad \text{for all }\bP^*\in\cP,
	\]
	implies that $\tilde V_s\in L^1(\Omega,\cF_s,\bP^*)$ for all $\bP^*\in\cP$ and all $s=0,\dots,T$. In particular, since $\tilde B$ is increasing and non-negative, we can conclude that $\tilde V$ is a $\P^*$-supermartingale for all $\bP^*\in\cP$. Furthermore, we show that $\tilde V_s\geq 0$ $\bP$-a.s. for all $s=0,\dots,T$. To this end, let $\bP^*\in\cP$ be arbitrary, then we have by the $\bP^*$-supermartingale property that
	\begin{align*}
	\tilde V_s&\geq \bE^*[\tilde V_T\mid\cF_s]\\
	&=\bE^*\left[\sup_{\bP^*\in\cP}\bE^*[H]+\sum_{k=1}^T\xi_k\cdot (X_k-X_{k-1})-\tilde B_T\mid\cF_s\right]\\
	&=\bE^*[H\mid\cF_s]\geq 0.
	\end{align*}
	The terminal value of $\tilde V$ dominates $H$ by construction and since $B_s\leq \tilde B_s$ for all $s=0,\dots,T$, we have
	\[
	\tilde V_s\leq \esssup_{\bP^*\in\cP}\bE^*[H\mid\cF_s]\quad\bP\text{-a.s. for all }s=0,1\dots,T.
	\]
	Then we obtain
	\[
	\bP(\tilde V_{t}<\esssup_{\bP^*\in\cP}\bE^*[H\mid\cF_t])=\bP(B_t<\tilde B_t)>0,
	\]
	which contradicts the fact that $(\esssup_{\bP^*\in\cP}\bE^*[H\mid\cF_s])_{s=0,\dots,T}$ is the smallest $\cP$-supermartingale whose terminal value dominates $H$. Thus $B_t= \tilde B_t$ $\bP$-a.s. This concludes the proof.
\end{proof}

\begin{remark}
	\label{rem:processofConsumptionID}
	In the definition of \eqref{eq:optimizationB} we can equivalently consider $\esssup \widehat\cB_t$, where
	\[
	\widehat \cB_t:=\left\{D_t\in L^0(\Omega,\cF_t,\bP): 0\leq D_t\leq \sup_{\bP^*\in\cP}\bE^*[H]+\sum_{k=1}^T\xi_k\cdot(X_k-X_{k-1})-H\ \bP\text{-a.s.}\right\},
	\]
	for $t=1,\dots,T$. This is due to the fact that, on the one hand $\cB_t\subset \widehat \cB_t$ for all $t=1,\dots,T$. On the other hand, for $D_t\in \widehat \cB_t$ we have that $\tilde D_t:=D_t\vee B_{t-1}\in\cB_t$ and $D_t\leq \tilde D_t$ $\bP$-a.s. Therefore, $\esssup \widehat \cB_t= \esssup\cB_t=B_t$ for all $t=1,\dots,T$.
\end{remark}

\subsection{Neural network approximation for $t>0$}
\label{sec:NNapproxt}
We now study a neural network approximation for the superhedging price process for $t>0$. Throughout this section we use the notation of Section \ref{sec:price0}. For $\varepsilon,\tilde\varepsilon\in (0,1)$ we define the set 
\begin{align*}
\cB_t^{\theta_t^*,\varepsilon,\tilde\varepsilon}:=\Bigg\lbrace &F^{\theta_t}(\cY_t):\, \theta_{t} \in \Theta_{t+1}\text{ and }\\
& \bP\left(B_{t-1}-\tilde\varepsilon\leq F^{\theta_t}(\cY_t)\leq \sup_{\bP^*\in\cP}\bE^*[H]+\sum_{k=1}^T\xi_k\cdot(X_k-X_{k-1})-H+\tilde\varepsilon\right)>1-\varepsilon\Bigg\rbrace,
\end{align*}
where $B$ is the consumption process for $H$ introduced in \eqref{eq:uniformdoob}. We now construct an approximation of $B$ by neural networks.

\begin{proposition}
\label{prop:NNconsumption}
	Assume $\sigma$ is bounded and non-constant. Then for any $\varepsilon,\tilde\varepsilon>0$ there exist neural networks $(F^{\theta_0,\varepsilon,\tilde\varepsilon},\dots,F^{\theta_T,\varepsilon,\tilde\varepsilon})$ such that $F^{\theta_t,\varepsilon,\tilde\varepsilon}(\cY_t)\in\cB_t^{\theta_t^*,\varepsilon,\tilde\varepsilon}$ for all $t=0,\dots,T$ and 
	\[
	\bP\left(\left|F^{\theta_t,\varepsilon,\tilde\varepsilon}(\cY_t)-B_t\right|>\tilde\varepsilon\right)<\varepsilon,\quad \text{for all }t=0,\dots,T.
	\]
	In particular, there exists a sequence of neural networks $\left(F^{\theta_0^n},\dots,F^{\theta_T^n}\right)_{n\in\bN}$ with $F^{\theta_t^n}(\cY_t)\in\cB_t^{\theta_t^*,\frac{1}{n},\frac{1}{n}}$ for all $n\in\bN$ and for all $t=0,\dots,T$ such that
\[
\left(F^{\theta_0^n}(\cY_0),\dots,F^{\theta_T^n}(\cY_T)\right)\xrightarrow{\bP\text{-a.s.}}(B_0,\dots,B_T)\quad \text{for }n\to\infty.
\]

\end{proposition}

\begin{proof}
	Fix $\varepsilon,\tilde\varepsilon>0$ and $t\in\{1,\dots,T\}$. Note that $B_0=0$ by definition. Let $B$ be given by the representation \eqref{eq:optimizationB}. Observe that the set $\cB_t$ from \eqref{eq:optimizationBset} is directed upwards. By Theorem A.33 of \cite{follmerschied} there exists an increasing sequence 
	\[
	(B_t^k)_{k\in\bN}\subset \cB_t,
	\]
	such that $B_t^k$ converges $\bP$-almost surely to $\tilde B_t=B_t$ as $k$ tends to infinity. Since almost sure convergence implies convergence in probability, there exists $K=K(\varepsilon,\tilde\varepsilon)\in\bN$ such that 
	\begin{equation}
	\label{eq:approxsequenceB}
	\bP\left(\left|B_t^k-B_t\right|> \frac{\tilde\varepsilon}{2}\right)<\frac{\varepsilon}{2},\quad\text{for all }k\geq K.
	\end{equation}
	For all $k\geq K$ there exist measurable functions $f_t^k:\bR^{mt}\to\bR$ such that $B_t^k=f_t^k(\cY_t)$. Fix $k\geq K$. By the universal approximation theorem \cite[Theorem~1 and Section~3]{hornik1991}, see also Theorem \ref{thm:universalapproxprob} in the appendix, (with measure $\mu$ given by the law of $\cY_t$ under $\P$) there exists $\theta_t=\theta_t^k\in\Theta_{t+1}$ and $F^{\theta_t}=F^{\theta_t^k,\varepsilon,\tilde\varepsilon}$ such that 
\[
\bP\left(\left|f_t^k(\cY_t)-F^{\theta_t}(\cY_t)\right|>\frac{\tilde\varepsilon}{2}\right)<\frac{\varepsilon}{2}.
\]
By the triangle inequality and by De Morgan's law we obtain that
\begin{align*}
&\left\lbrace \omega\in\Omega:\left|B_t(\omega)-F^{\theta_t}(\cY_t(\omega))\right|>\tilde\varepsilon\right\rbrace\subseteq \left\lbrace \omega\in\Omega:\left|B_t(\omega)-B_t^k(\omega)\right|+\left|B_t^k-F^{\theta_t}(\cY_t(\omega))\right|>\tilde\varepsilon\right\rbrace\\
\subseteq&\left\lbrace\omega\in\Omega:\left|B_t(\omega)-B_t^k(\omega)\right|>\frac{\tilde\varepsilon}{2}\right\rbrace
\cup\left\lbrace\omega\in\Omega:\left|B_t^k(\omega)-F^{\theta_t}(\cY_t(\omega))\right|>\frac{\tilde\varepsilon}{2}\right\rbrace.
\end{align*}
In particular, we obtain by sub-addidivity that
\[
\bP\left(\left|B_t-F^{\theta_t}(\cY_t)\right|>\tilde\varepsilon\right)\leq \bP\left(\left|B_t-B_t^k\right|>\frac{\tilde\varepsilon}{2}\right)+\bP\left(\left|B_t^k-F^{\theta_t}(\cY_t)\right|>\frac{\tilde\varepsilon}{2}\right)<\frac{\varepsilon}{2}+\frac{\varepsilon}{2}=\varepsilon.
\]
Next, we show that $F^{\theta_t}\in\cB_t^{\theta_t^*,\varepsilon,\tilde\varepsilon}$. For this purpose, we note that 
\[
B_{t-1}\leq B_t\leq \sup_{\bP^*\in\cP}\bE^*[H]+\sum_{k=1}^T\xi_k\cdot\left(X_k-X_{k-1}\right)-H\quad\bP\text{-a.s.}
\]
Therefore, we have that
\[
\resizebox{15cm}{!}
{
$\bP\left(B_{t-1}-\tilde\varepsilon\leq F^{\theta_t}(\cY_t)\leq \sup_{\bP^*\in\cP}\bE^*[H]+\sum_{k=1}^T\xi_k\cdot\left(X_k-X_{k-1}\right)-H+\tilde\varepsilon\right)\geq \bP\left(\left|B_t-F^{\theta_t}(\cY_t)\right|\leq\tilde\varepsilon\right)>1-\varepsilon,$}
\]
which implies that $F^{\theta_t}(\cY_t)=F^{\theta_t^k,\varepsilon,\tilde\varepsilon}(\cY_t)\in\cB_t^{\theta_t^*,\varepsilon,\tilde\varepsilon}$. We set $\varepsilon=\frac{1}{n}=\tilde\varepsilon$ for $n\in\bN$ and consider the neural network
\[
F^{\theta_t^n}:=F^{\theta_t^{K(n)},\frac{1}{n},\frac{1}{n}},\quad t\in\{1,\dots,T\},\ n\in\bN,
\]
where $K(n)=K(\frac{1}{n},\frac{1}{n})$ is given by \eqref{eq:approxsequenceB}. Then, $F^{\theta_t^n}\in\cB_t^{\theta_t^*,\frac{1}{n},\frac{1}{n}}$ for all $n\in\bN$ and for all $t=1,\dots,T$. Further, we have
\[
\bP\left(\left|F^{\theta_t^n}(\cY_t)-B_t\right|>\frac{1}{n}\right)<\frac{1}{n}\quad \text{for all }t=1,\dots,T,
\]
which implies convergence in probability, i.e.,
\[
F^{\theta_t^n}(\cY_t)\xrightarrow{\bP}B_t\quad\text{for }n\to\infty,\text{ for all }t=0,\dots,T.
\]
By passing to a suitable subsequence, convergence also holds $\bP$-a.s. simultaneously for all $t=0,\dots,T$.
\end{proof}\noindent
Let $\tilde\varepsilon>0$. Recursively, we define the set
\begin{equation}\begin{aligned}
\label{eq:NNapproxTset}
\tilde\cB_t^{\theta_t^*,\tilde\varepsilon}:= \large\lbrace  & F^{\theta_t}(\cY_t)\mathds{1}_A + B_{t-1}^{\theta_{t-1}^*,\tilde\varepsilon}\mathds{1}_{A^c} :\, \theta_{t} \in \Theta_{t+1}, A \in \Fc_t,  \\ & B_{t-1}^{\theta_{t-1,\tilde\varepsilon}^*}\leq F^{\theta_t}(\cY_t)\mathds{1}_A + B_{t-1}^{\theta_{t-1}^*,\tilde\varepsilon}\mathds{1}_{A^c} 
\leq \sup_{\bP^*\in\cP}\bE^*[H]+\sum_{k=1}^T\xi_k\cdot(X_k-X_{k-1})-H+\tilde\varepsilon\large\rbrace 
\end{aligned}
\end{equation}
for $t=1,\dots,T$, and the approximated process of consumption by $B_0^{\theta_0^*,\tilde\varepsilon}=0$ and 
\begin{equation}\begin{aligned}
\label{eq:NNapproxTprocess}
    B_t^{\theta_t^*,\tilde\varepsilon}:=
    \esssup \tilde\cB_t^{\theta_t^*,\tilde\varepsilon}\quad \text{ for }t=1,\dots,T.  
\end{aligned}
\end{equation}

\begin{theorem}
\label{thm:NNgeneralt}
	Assume $\sigma$ is bounded and non-constant. Then 
	\[
	\left|B_t^{\theta_t^*,\tilde\varepsilon}-B_t\right|\leq\tilde\varepsilon \quad\bP\text{-a.s. for all }t=0,\dots,T.
	\]
\end{theorem}

\begin{proof}
	We prove the statement by induction. For $t=0$ we have by definition $B_0^{\theta_0^*,\tilde\varepsilon}=B_0=0$. Assume now that 
	\[
	\left|B_{t-1}^{\theta_{t-1}^*,\tilde\varepsilon}-B_{t-1}\right|\leq\tilde\varepsilon \quad\bP\text{-a.s.}
	\] 
	for some $t\in \{1,\dots,T\}$. First we note that $B_s^{\theta_s^*,\tilde\varepsilon}\leq B_{s+1}^{\theta_{s+1}^*,\tilde\varepsilon}$ by \eqref{eq:NNapproxTset} and \eqref{eq:NNapproxTprocess}, and because $B_0^{\theta_0^*,\tilde\varepsilon}=0$ it follows that $B_s^{\theta_s^*,\tilde\varepsilon}\geq 0$ for all $s=1,\dots,T$. Let $\theta_t\in \Theta_{t+1}$ and $A\in\cF_t$ such that 
	\[
	0\leq B_{t-1}^{\theta_{t-1}^*,\tilde\varepsilon}\leq F^{\theta_t}(\cY_t)\mathds{1}_A + B_{t-1}^{\theta_{t-1}^*,\tilde\varepsilon}\mathds{1}_{A^c} \leq \sup_{\bP^*\in\cP}\bE^*[H]+\sum_{k=1}^T\xi_k\cdot(X_k-X_{k-1})-H+\tilde\varepsilon.
	\]
	Then, we can easily see that
	\[
	\resizebox{15cm}{!}
	{
	$F^{\theta_t}(\cY_t)\mathds{1}_A + B_{t-1}^{\theta_{t-1}^*,\tilde\varepsilon}\mathds{1}_{A^c}\in \left\lbrace D_t\in L^0(\Omega,\cF_t,\bP):0 \leq D_t\leq \sup_{\bP^*\in\cP}\bE^*[H]+\sum_{k=1}^T\xi_k\cdot(X_k-X_{k-1})-H+\tilde\varepsilon \ \bP\text{-a.s.}\right\rbrace$}.
	\]
	We now prove that
	\begin{equation}
	\label{eq:NNdir1B}
	\resizebox{15cm}{!}
	{$ B_t+\tilde\varepsilon=\esssup\left\lbrace D_t\in L^0(\Omega,\cF_t,\bP):-\tilde\varepsilon \leq D_t-\tilde\varepsilon\leq \sup_{\bP^*\in\cP}\bE^*[H]+\sum_{k=1}^T\xi_k\cdot(X_k-X_{k-1})-H\ \bP\text{-a.s.}\right\rbrace$}.
	\end{equation}
	On the one hand we have
	\begin{align*}
	&\left\lbrace \tilde D_t\in L^0(\Omega,\cF_t,\bP): 0\leq \tilde D_t\leq \sup_{\bP^*\in\cP}\bE^*[H]+\sum_{k=1}^T\xi_k\cdot(X_k-X_{k-1})-H\ \bP\text{-a.s.}\right\rbrace+\tilde\varepsilon\\
		=&\left\lbrace D_t\in L^0(\Omega,\cF_t,\bP): 0\leq D_t-\tilde\varepsilon\leq \sup_{\bP^*\in\cP}\bE^*[H]+\sum_{k=1}^T\xi_k\cdot(X_k-X_{k-1})-H\ \bP\text{-a.s.}\right\rbrace\\
		\subseteq &\left\lbrace D_t\in L^0(\Omega,\cF_t,\bP):-\tilde\varepsilon \leq D_t-\tilde\varepsilon\leq \sup_{\bP^*\in\cP}\bE^*[H]+\sum_{k=1}^T\xi_k\cdot(X_k-X_{k-1})-H\ \bP\text{-a.s.}\right\rbrace,
	\end{align*}
	which by Remark \ref{rem:processofConsumptionID} implies that 
	\[
	\resizebox{15cm}{!}
	{$ B_t+\tilde\varepsilon\leq \esssup\left\lbrace D_t\in L^0(\Omega,\cF_t,\bP):-\tilde\varepsilon \leq D_t-\tilde\varepsilon\leq \sup_{\bP^*\in\cP}\bE^*[H]+\sum_{k=1}^T\xi_k\cdot(X_k-X_{k-1})-H\ \bP\text{-a.s.}\right\rbrace$}.
	\]
	On the other hand, let 
	\[
	D_t\in \left\lbrace D_t\in L^0(\Omega,\cF_t,\bP):-\tilde\varepsilon \leq D_t-\tilde\varepsilon\leq \sup_{\bP^*\in\cP}\bE^*[H]+\sum_{k=1}^T\xi_k\cdot(X_k-X_{k-1})-H\ \bP\text{-a.s.}\right\rbrace,
	\]
	and define $\tilde D_t:=D_t\vee \tilde\varepsilon$. Then $D_t\leq \tilde D_t$ $\bP$-a.s. and
	\[
	\tilde D_t\in \left\lbrace \bar D_t\in L^0(\Omega,\cF_t,\bP):0 \leq \bar D_t-\tilde\varepsilon\leq \sup_{\bP^*\in\cP}\bE^*[H]+\sum_{k=1}^T\xi_k\cdot(X_k-X_{k-1})-H\ \bP\text{-a.s.}\right\rbrace,
	\]
	which implies that
	\[
	\resizebox{15cm}{!}
	{$ B_t+\tilde\varepsilon\geq \esssup\left\lbrace D_t\in L^0(\Omega,\cF_t,\bP):-\tilde\varepsilon \leq D_t-\tilde\varepsilon\leq \sup_{\bP^*\in\cP}\bE^*[H]+\sum_{k=1}^T\xi_k\cdot(X_k-X_{k-1})-H\ \bP\text{-a.s.}\right\rbrace$},
	\]
	and hence \eqref{eq:NNdir1B} follows. Further, we also have that
	\begin{align*}
		&\left\lbrace D_t\in L^0(\Omega,\cF_t,\bP):0 \leq D_t\leq \sup_{\bP^*\in\cP}\bE^*[H]+\sum_{k=1}^T\xi_k\cdot(X_k-X_{k-1})-H+\tilde\varepsilon \ \bP\text{-a.s.}\right\rbrace\\
		=&\left\lbrace D_t\in L^0(\Omega,\cF_t,\bP):-\tilde\varepsilon \leq D_t-\tilde\varepsilon\leq \sup_{\bP^*\in\cP}\bE^*[H]+\sum_{k=1}^T\xi_k\cdot(X_k-X_{k-1})-H\ \bP\text{-a.s.}\right\rbrace.
	\end{align*}
	Therefore, we obtain by \eqref{eq:NNdir1B} that
	\[
	F^{\theta_t}(\cY_t)\mathds{1}_A + B_{t-1}^{\theta_{t-1}^*,\tilde\varepsilon}\mathds{1}_{A^c}\leq B_t+\tilde\varepsilon\quad \bP\text{-a.s.,}
	\]
	and hence
	\begin{equation}
	\label{eq:NNconsumptiondir1}
	B_t^{\theta_t^*,\tilde\varepsilon}\leq B_t+\tilde\varepsilon\quad\bP\text{-a.s.}
	\end{equation}
	For the converse direction let $\varepsilon\in (0,1)$. By the proof of Proposition \ref{prop:NNconsumption} there exists a neural network $F^{\tilde\theta_t}=F^{\tilde\theta_t,\varepsilon,\tilde\varepsilon}$ such that 
	\[
	\bP\left(\left|F^{\tilde\theta_t}(\cY_t)-B_t\right|>\tilde\varepsilon\right)<\varepsilon.
	\]
	Define the sets $A_1,A_2\in\cF_t$ by
	\[
	A_1:=\left\lbrace \omega\in\Omega: B_t(\omega)-\tilde\varepsilon \leq F^{\tilde\theta_t}(\cY_t(\omega))\leq B_t(\omega)+\tilde\varepsilon\right\rbrace,
	\]
	and
	\[
	A_2:=\left\lbrace \omega\in\Omega: B_{t-1}^{\theta_{t-1}^*,\tilde\varepsilon}(\omega)\leq F^{\tilde\theta_t}(\cY_t(\omega))\right\rbrace.
	\]
	Then, $\bP(A_1)> 1-\varepsilon$. Note that by the assumption of the induction
	\[
	B_{t-1}^{\theta_{t-1}^*,\tilde\varepsilon}\leq B_{t-1}+\tilde\varepsilon\leq B_t+\tilde\varepsilon\quad\bP\text{-a.s.}
	\]
	For $A:=A_1\cap A_2$ we have by construction,
	\[
	F^{\tilde\theta_t}(\cY_t)\mathds{1}_A+B_{t-1}^{\theta_{t-1}^*,\tilde\varepsilon}\mathds{1}_{A^c}=F^{\tilde\theta_t}(\cY_t)\mathds{1}_{A_1\cap A_2}+B_{t-1}^{\theta_{t-1}^*,\tilde\varepsilon}\mathds{1}_{A_1^c\cup A_2^c}\in \tilde\cB_t^{\theta_t^*,\tilde\varepsilon}.
	\]
	For $\omega\in A_1\cap A_2^c$ we get  that
	\[
	F^{\tilde\theta_t}(\cY_t(\omega))\mathds{1}_{A_1\cap A_2}(\omega)+B_{t-1}^{\theta_{t-1}^*,\tilde\varepsilon}(\omega)\mathds{1}_{A_1^c\cup A_2^c}(\omega)=B_{t-1}^{\theta_{t-1}^*,\tilde\varepsilon}(\omega)
	\] 
	and 	\[
	B_t(\omega)-\tilde\varepsilon\leq F^{\tilde\theta_t}(\cY_t(\omega))< B_{t-1}^{\theta_{t-1}^*,\tilde\varepsilon}(\omega)\leq B_{t}(\omega)+\tilde\varepsilon.
	\]
	For $\omega\in A_1\cap A_2$ we have
	\[
	F^{\tilde\theta_t}(\cY_t(\omega))\mathds{1}_{A_1\cap A_2}(\omega)+B_{t-1}^{\theta_{t-1}^*,\tilde\varepsilon}(\omega)\mathds{1}_{A_1^c\cup A_2^c}(\omega)=F^{\tilde\theta_t}(\cY_t(\omega))
	\]
	and 	
	\[
	\left| F^{\tilde\theta_t}(\cY_t(\omega))- B_t(\omega)\right|\leq \tilde\varepsilon.
	\] 
	Thus, using that that $A_1=(A_1\cap A_2)\cup (A_1\cap A_2^c)$ and $\bP(A_1)>1-\varepsilon$ we get
	\begin{equation}
	\label{eq:NNconsumptionProbsets}
	\bP\left(\left|\left(F^{\tilde\theta_t}(\cY_t)\mathds{1}_A+B_{t-1}^{\theta_{t-1}^*,\tilde\varepsilon}\mathds{1}_{A^c}\right)-B_t\right|>\tilde\varepsilon\right)\leq \bP(A_1^c)<\varepsilon.
	\end{equation}
	Then, \eqref{eq:NNconsumptionProbsets} implies
	\begin{equation}
	\label{eq:NNconsumptionconclusion}
	\bP\left(B_t^{\theta_t^*,\tilde\varepsilon}<B_t-\tilde\varepsilon\right)\leq \bP\left(F^{\tilde\theta_t}(\cY_t)\mathds{1}_A+B_{t-1}^{\theta_{t-1}^*,\tilde\varepsilon}\mathds{1}_{A^c}<B_t-\tilde\varepsilon\right)<\varepsilon.
	\end{equation}
	Because $\varepsilon\in (0,1)$ was arbitrary, it follows that $B_t\leq B_t^{\theta_t^*,\tilde\varepsilon}+\tilde\varepsilon$ $\bP$-a.s. by \eqref{eq:NNconsumptionconclusion}. By \eqref{eq:NNconsumptiondir1} and \eqref{eq:NNconsumptionconclusion} we conclude that $|B_t^{\theta_t^*,\tilde\varepsilon}-B_t|\leq\tilde\varepsilon$ $\bP$-a.s. for all $t=0,\dots,T$.
\end{proof}

\section{Numerical results}
\label{sec:numericalresults}

In this section, we present some numerical applications for the results in Section \ref{sec:price0} and \ref{sec:superhedgingt}. Combining Theorem \ref{thm:convergencesuccessset} and \ref{thm:universalapprox2}, we obtain a two-step approximation for the superhedging price at $t=0$. Then, we use Theorem \ref{thm:NNgeneralt} to simulate the superhedging process for $t>0$.
\subsection{Case $t=0$}
\label{sec:case0}
\subsubsection{Algorithm and implementation}
\label{sec:algorithm}

Let $N\in\bN$ denote a fixed batch size. For fixed $\lambda>0$ we implement the following iterative procedure: for each iteration step $i$ we generate i.i.d.\ samples $Y(\omega_{0}^{(i)}), \ldots, Y(\omega_{N}^{(i)})$  of $Y$ and consider the empirical loss function
\begin{align*}
L_\lambda^{(i)}(\theta) = &\left|F^{\theta_u}\left(\Yc_0\left(\omega_{0}^{(i)}\right)\right)\right|^2 + \frac{\lambda}{N} \sum_{j=1}^N l\Bigg(H\left(\omega_{j}^{(i)}\right)\\
&-\left[F^{\theta_u}\left(\Yc_0\left(\omega_{j}^{(i)}\right)\right) 
+ \sum_{k=1}^T F^{\theta_{k,\xi}}\left(\Yc_{k-1}\left(\omega_{j}^{(i)}\right)\right) \cdot \left(X_k\left(\omega_{j}^{(i)}\right) - X_{k-1}\left(\omega_{j}^{(i)}\right)\right) \right]\Bigg),
\end{align*} 
with $\theta=(\theta_u,\theta_{1,\xi},\dots,\theta_{T,\xi})$ and $l:\bR\to [0,\infty)$ denoting the squared \emph{rectifier} function, i.e.,
\[
l(x)=\left(\max\left\{x,0\right\}\right)^2.
\]
We then calculate the gradient of $L_\lambda^{(i)}(\theta)$ at $\theta^{(i)}$ and use it to update the parameters from $\theta^{(i)}$ to $\theta^{(i+1)}$ according to the \emph{Adam} optimizer, see \cite{kingma2014adam}. After sufficiently many iterations $i$, the parameter $\theta^{(i)}$ should be sufficiently close to a local minimum of the loss function
\begin{equation}\label{eq:loss}
L_\lambda(\theta) = \left|F^{\theta_u}\left(\Yc_0\right)\right|^2 + \lambda \E\left[ l\left(H-\left(F^{\theta_u}\left(\Yc_0\right) + \sum_{k=1}^T F^{\theta_{k,\xi}}\left(\Yc_{k-1}\right) \cdot \left(X_k - X_{k-1}\right) \right)\right) \right].
\end{equation}
Note that $\cY_0$ is constant and hence $F^{\theta_u}(\cY_0)$ is a constant. We obtain a small value for the first term of $L_\lambda$ if $F^{\theta_u}(\cY_0)$ representing the superhedging price is small. On the other side, the second summand in \eqref{eq:loss} is equal $0$ when the portfolio dominates the claim $H$. Thus, minimizing the second summand in \eqref{eq:loss} corresponds to maximizing the superhedging probability. The weight $\lambda$ offers the opportunity to balance between a small initial price of the portfolio and a high probability of superhedging. In particular, if $\theta$ is the minimum for the loss function $L_\lambda(\theta)$, then $F^{\theta_u}(\cY_0)$ is close to the minimal price required to superhedge the claim $H$ with a certain probability, i.e., to the quantile hedging price for a certain $\alpha=\alpha(\lambda)$. In view of Theorem \ref{thm:universalapprox2} we thus expect $F^{\theta_u}(\cY_0)\approx \inf\cU_0$ for $\lambda$ large enough.\\
Also other choices for $l$ in \eqref{eq:loss} are possible. We considered the scaled \textit{sigmoid} function for $l$ in \eqref{eq:loss}. In this case, however, we did not obtain stable results.\\
The algorithm is implemented in Python, using Keras with backend TensorFlow to build and train the neural networks. More precisely, we create a \emph{Sequential} object to build the models and compile with a customized loss function.\\
We use a Long-Short-Term-Memory network (LSTM), see \cite{hochreiter1997lstm}, with the following architecture: the network has two LSTM layers of size $30$, which return sequences and one dense layer of size $1$. Between the layers the \emph{swish} activation function is used. The activation functions within the LSTM layers are set to default, i.e., activation between cells is \emph{tanh} and the recurrent activation is the \emph{sigmoid} function. The kernel and bias initializer of the first LSTM layer are set to \emph{truncated normal}, i.e., the initial weights are drawn from a standard normal distribution but we discard and re-draw values, which are more than two standard deviations from the mean. This gives $11191$ trainable parameters. The training is performed using the \emph{Adam} optimizer with a learning rate of $0.001$ or $0.0001$. We generate $1024000$ samples, which we split in $70\%$ for the training set and $30\%$ for the test set. The batch size is set to $1024$. We apply the procedure described above in two examples, which we present in the following.

\subsubsection{Trinomial model}
\label{sec:trinomial}

We consider a discrete time financial market model given by an arbitrage-free trinomial model with $X_0=100$ and
\[
X_t=X_0\prod_{k=1}^t (1+R_t), \quad t\in \{0,\dots,T\},
\]
where $R_t$ is $\cF_t$-measurable for $t\in \{1,\dots,T\}$, and takes values in $\{d,m,u\}$ with equal probability, where $-1<d<m<u$. Here, we set $d=-0.01$, $m=0$, and $u=0.01$ and $T=29$ yielding $3^{29}$ possible paths. In this model, we want to superhedge a European Call option $H=(X_T-K)^+$ with strike price $K=100$. For this choice of parameters the theoretical superhedging price is $2.17$, as it can be easily obtained by the results of \cite{carassus2010super}. \\
The network is trained and evaluated for different $\lambda$ to illustrate the impact of $\lambda$ in \eqref{eq:loss} and the relation between $\alpha(\lambda)\in (0,1)$ and the corresponding $\alpha(\lambda)$-quantile hedging price. More precisely, we consider $\lambda\in \{10, 50, 100, 500, 1000, 2000, 4000, 10000\}$. For each $\lambda$ the network is trained over $40$ epochs.\\
In Figure \ref{fig:impactlambda_fin}(a)-(c), we see that $\alpha(\lambda)$ as well as the $\alpha(\lambda)$-quantile hedging price increase in $\lambda$, and that the $\alpha(\lambda)$-quantile hedging price increases in $\alpha(\lambda)$. Figure~\ref{fig:impactlambda_fin}(d) shows the superhedging performance on the test set for all $\lambda$'s, i.e., samples of
\begin{equation}
\label{eq:superhedgingperformance}
F^{\theta_u(\lambda)}\left(\Yc_0\right) + \sum_{k=0}^T F^{\theta_{k,\xi}(\lambda)}\left(\Yc_{k-1}\right) \cdot \left(X_k - X_{k-1}\right) -H,
\end{equation}
for each $\lambda$. Table \ref{tab:lambda_fin} summarizes the values for $\lambda$, $\alpha(\lambda)$ and the $\alpha(\lambda)$-quantile hedging price. In particular, for $\lambda=10000$ we obtain a numerical price of $2.15$ and $\alpha(\lambda)=99.24\%$. 

\begin{figure}[htbp]
\begin{subfigure}[$\alpha(\lambda)$-quantile hedging price depending on $\lambda$]{\includegraphics[width=0.49\textwidth]{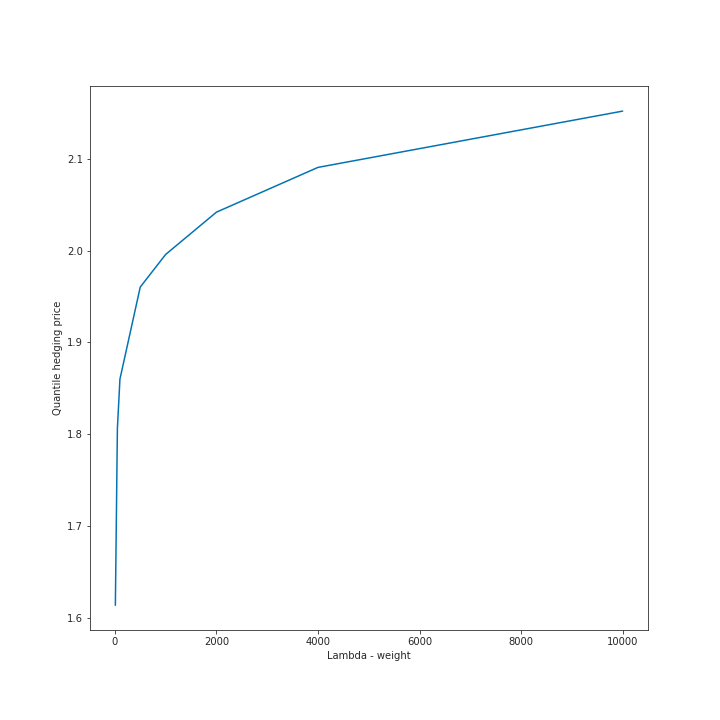}}
\label{subfig:price_weight}
\end{subfigure}
\begin{subfigure}[$\alpha(\lambda)$ depending on $\lambda$]{\includegraphics[width=0.49\textwidth]{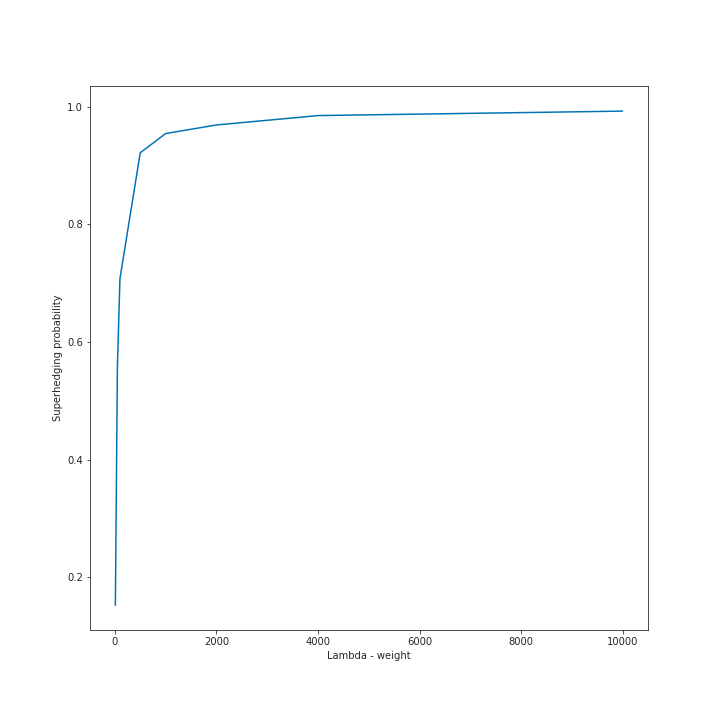}}
\end{subfigure}
\begin{subfigure}[$\alpha(\lambda)$-quantile hedging price depending on $\alpha(\lambda)$]{\includegraphics[width=.49\textwidth]{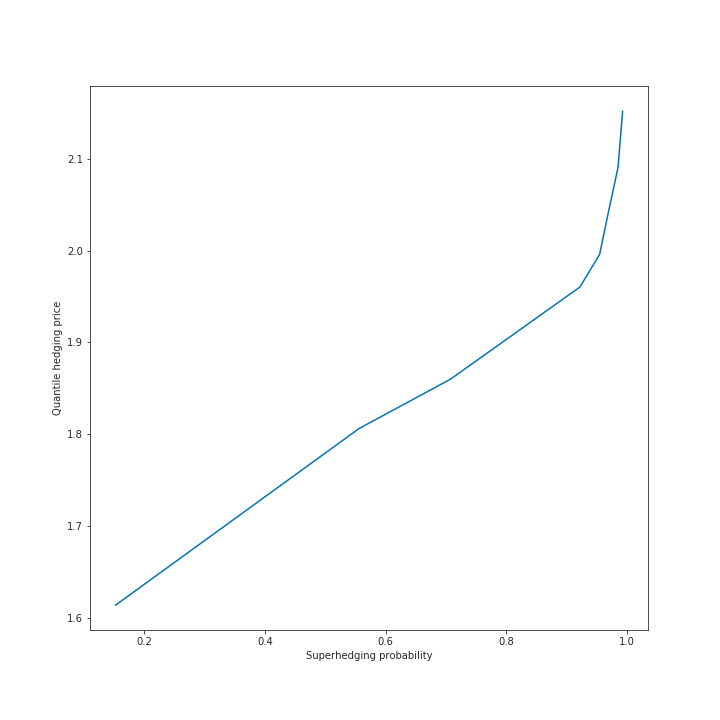}}
\label{subfig:price_prob}	
\end{subfigure}
\begin{subfigure}[Superhedging performance]{\includegraphics[width=.49\textwidth]{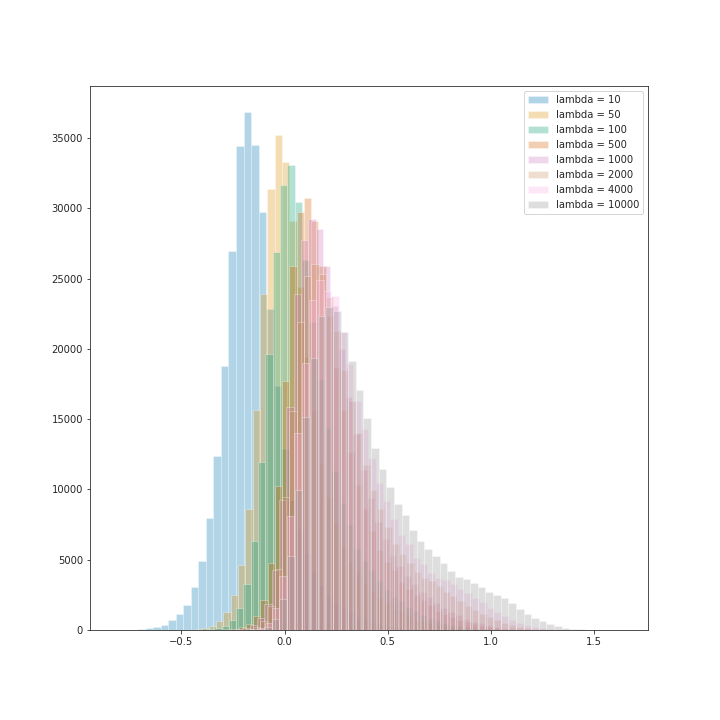}}
\label{subfig:price_prob}	
\end{subfigure}

\caption{Impact of $\lambda$ on the quantile hedging price and on the superhedging probability.}
\label{fig:impactlambda_fin}
\end{figure}

\begin{table}[hbt]
\begin{center}
\begin{tabular}{ |l|c|r| }
\hline
$\lambda$ & $\alpha(\lambda)$ & $\alpha(\lambda)$-quantile hedging price \\
\hline
$10$ & $15.23\%$ & 	$1.61$ \\
\hline
$50$ & $55.61\%$ & 	$1.81$ \\
\hline
$100$ & $70.75\%$ & $1.86$ \\
\hline
$500$ & $92.16\%$ & $1.96$ \\
\hline
$1000$ & $95.42\%$ & $2.00$ \\
\hline
$2000$ & $96.88\%$ & $2.04$ \\
\hline
$4000$ & $98.48\%$ & $2.09$ \\
\hline
$10000$ & $99.24\%$ & $2.15$ \\
\hline
\end{tabular}
\end{center}

\caption{Impact of $\lambda$ on $\alpha(\lambda)$ and on the $\alpha(\lambda)$-quantile hedging price.}
\label{tab:lambda_fin}
\end{table}

\subsubsection{Discretized Black Scholes model}
\label{sec:blackscholes}

Here we consider a discrete time financial market given by a discretized Black-Scholes model for the asset price $X$. We consider a Barrier Up and Out Call option $H=\prod_{t=0}^T \mathds{1}_{\{X_t<U\}} (X_T-K)^+$ with strike $K=100$ and upper bound $U=105$ such that $K<U$ and $X_0<U$. We set $X_0=100$, $\sigma=0.3$ and $\mu=0$.  We assume to have $250$ trading days per year and a time horizon $T$ of $30$ trading days with daily rebalancing. In particular, for a European contingent claim the time until expiration for the option is $\tau=30/250$. \\
The weight $\lambda$ of the loss function is set to $10000000$ in order to obtain a high superhedging probability. Indeed, we obtain a superhedging probability of $100\%$ on the training set as well as on the test set with an approximate price of $3.73$. By \cite{carassus2007class}, the theoretical superhedging price $\pi^H$ is given by
\[
\pi^H=X_0\left(1-\frac{K}{U}\right)\approx 4.76.
\]
In the Black-Scholes model the asset price process at time $t>0$ has unbounded support and thus the additional error, which arises from the discretization of the probability space, is non-negligible. Although the Barrier option artificially bounds the support of the model, the numerical price still significantly deviates from the theoretical price. \\
Finally, we consider a European call option $H=(X_T-K)^+$ with strike $K=100$ and parameters $X_0=100$, $\sigma=0.1$ and $\mu=0$. By \cite{carassus2007class} the theoretical price of $H$ for the discrete time version of the Black-Scholes model is equal to $X_0=100$. The theoretical price of $H$ in a standard Black-Scholes model in the continuous time is $1.38$, and by following the $\delta$-hedging strategy we superhedge $H$ with a probability of $53.69\%$. Here we consider $\lambda=50$ in \eqref{eq:loss} in order to compare the result to the discretized $\delta$-hedging strategy of the Black-Scholes model, and $\lambda = 10000$ in order to obtain a high superhedging probability. For $\lambda=50$, we obtain an approximate price of $1.41$ and a superhedging probability of $54.43\%$. In Figure \ref{fig:bscall}(a) we compare the $\delta$-hedging strategy with the approximated superhedging strategy obtained for $\lambda=50$. Further, in Figure \ref{fig:bscall}(b) we compare the results for $\lambda=50$ and $\lambda=10000$, respectively. For $\lambda=10000$, the superhedging probability on the test set is $99.79\%$ with an approximated price of $2.18$. 

\begin{figure}[htbp]
\begin{subfigure}[$\delta$-hedging strategy compared to approximate strategy for $\lambda=50$]{\includegraphics[width=0.49\textwidth]{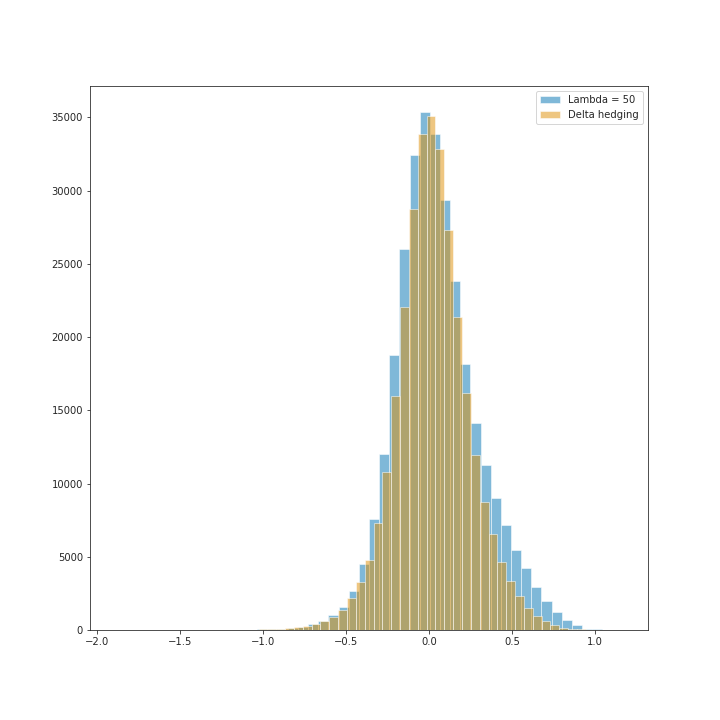}}
\label{subfig:price_weight}
\end{subfigure}
\begin{subfigure}[Approximate strategy for $\lambda=50$ and $\lambda=10000$]{\includegraphics[width=0.49\textwidth]{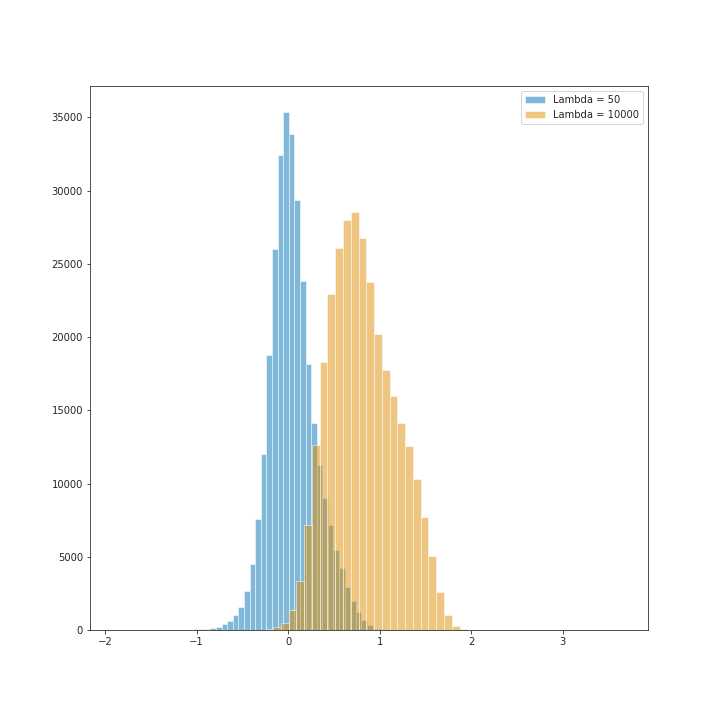}}
\label{subfig:prob_weight}
\end{subfigure}
\caption{Hedging losses for $\lambda=50$, $\lambda=10000$ and for the $\delta$-hedging strategy.}
\label{fig:bscall}
\end{figure}

\subsection{Case $t>0$}
\label{sec:caset}

In this section we approximate the process of consumption by neural networks as proposed in Section \ref{sec:NNapproxt}. We implement the same iterative procedure as introduced in Section \ref{sec:algorithm}. We define $G^{(i)}$ as the difference of the approximated superhedging strategy obtained from Section \ref{sec:case0} and the claim $H$, i.e.,
\[
\resizebox{15cm}{!}{
$G_j^{(i)}(\theta^*):=\left[F^{\theta_u^*}\left(\Yc_0\left(\omega_{j}^{(i)}\right)\right) 
+ \sum_{k=1}^T F^{\theta_{k,\xi}^*}\left(\Yc_{k-1}\left(\omega_{j}^{(i)}\right)\right) \cdot \left(X_k\left(\omega_{j}^{(i)}\right) - X_{k-1}\left(\omega_{j}^{(i)}\right)\right) -H\left(\omega_{j}^{(i)}\right)\right].$}
\]
Then, the empirical loss function is given by
\[
\tilde L_{t,\beta}^{(i)}(\theta_t)=\frac{1}{N}\sum_{j=1}^N-\left|B_t^{\theta_t}\left(\omega_j^{(i)}\right)\right|^2+\beta \max\left\{\left(B_t^{\theta_t}\left(\omega_j^{(i)}\right)-G_j^{(i)}(\theta^*)\right),0\right\},
\]
where $B_t^{\theta_t}$ is given by 
\[
B_t^{\theta_t}\left(\omega_j^{(i)}\right):=\max\left\{F^{\theta_t}\left(\cY_t\left(\omega_j^{(i)}\right)\right),B_{t-1}^{\theta_{t-1}}\left(\omega_j^{(i)}\right)\right\}.
\]
At a local minimum the two terms of $\tilde L$ guarantee that $F^{\theta_t}$ is as big as possible but less or equal than $G(\theta^*)$.\\
Here, we also consider a discretized Black-Scholes model as in Section \ref{sec:blackscholes} but only a time horizon of $10$ trading days and set $X_0=100$, $\sigma=0.1$ and $\mu=0$. For each $t>0$ the neural network consists of two LSTM layers of size $30$ and $20$ respectively, which return sequences, one LSTM layer of size $20$ providing one single value and one dense layer of size $1$. The remaining parameters are chosen as in Section \ref{sec:algorithm}.\\
As in Section \ref{sec:blackscholes}, we compute an approximated superhedging price and strategy for the complete interval. Setting $\lambda=1024$ yields an approximated price of $1.35$ and a superhedging probability of $98.87\%$ for $t=0$. For $t\geq 1$, we choose $\beta=500$ and then obtain a superhedging probability of $98.78\%$. In Figure \ref{fig:price_process}(a), we show trajectories of the approximated superhedging price process generated by this method. Figure \ref{fig:price_process}(b) illustrates paths given by the $\delta$-hedging strategy of the discretized Black-Scholes model. Finally, we plot the difference of the approximated superhedging price processes and the corresponding price process obtained by the $\delta$-hedging strategy in Figure \ref{fig:price_process}(c).

\begin{figure}[htbp]
\begin{subfigure}[Superhedging price process]{\includegraphics[width=0.49\textwidth]{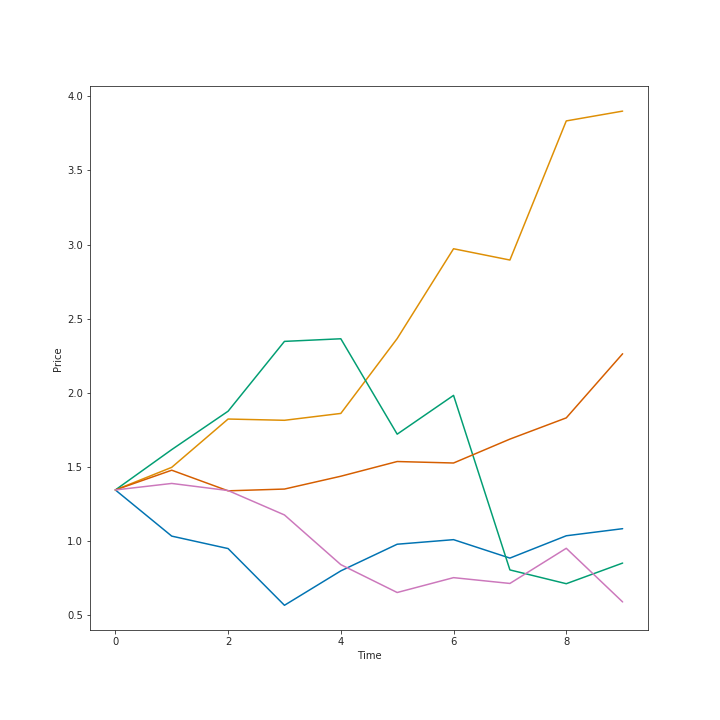}}
\label{subfig:pricepath}
\end{subfigure}
\begin{subfigure}[$\delta$-hedging price process]{\includegraphics[width=0.49\textwidth]{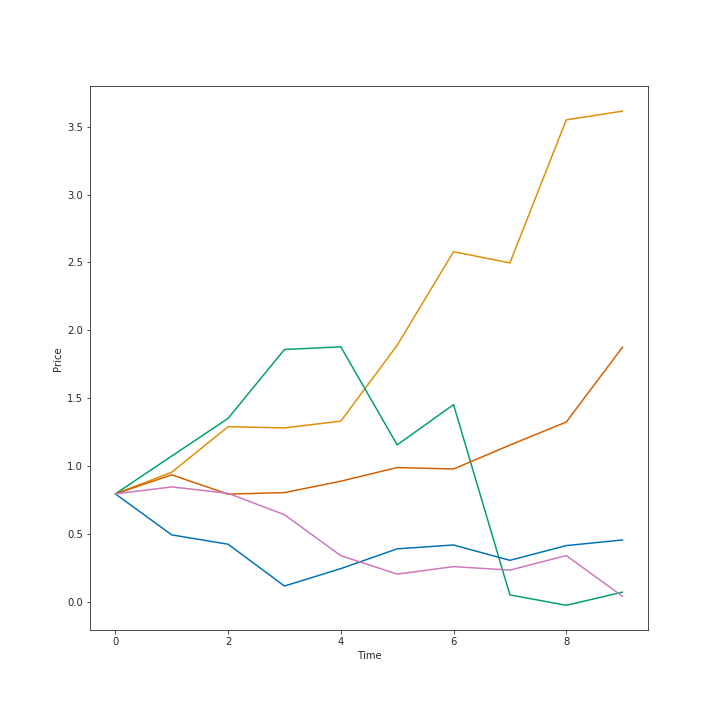}}
\label{subfig:delta_process}
\end{subfigure}
\begin{center}
\begin{subfigure}[Difference of the price processes]{\includegraphics[width=.6\textwidth]{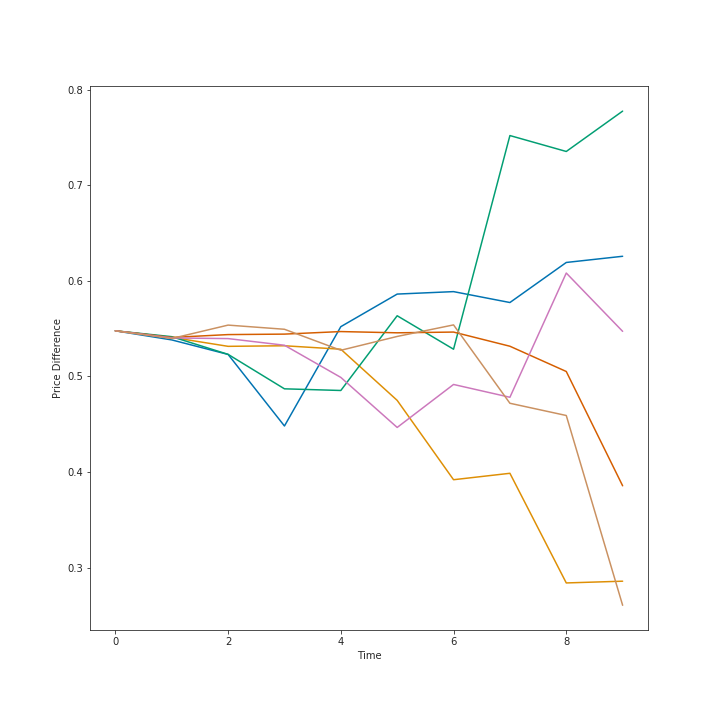}}
\label{subfig:price_diff}	
\end{subfigure}
\end{center}
\caption{Superhedging price process compared to the $\delta$-hedging price process.}

\label{fig:price_process}
\end{figure}

\subsection{Discussion}

In finite market models as in Section \ref{sec:trinomial}, our methodology delivers an approximation of $\alpha$-quantile hedging and approximated superhedging prices with small approximation error. It is also worth noting, that the predicted superhedging price and the corresponding superhedging probability of the training set are consistent with the values on the test set.\\
In contrast, in models in which the price process has unbounded support, our numerical results indicate that the additional error caused by the discretization of the probability space cannot be ignored. However, we obtain consistent results of the $\alpha$-quantile hedging price for the training set and test set. Note also that, in Section \ref{sec:blackscholes}, the Barrier option can be superhedged with $100\%$ on the training and on the test set. \\
A further possible application of our methodology is given by superhedging in a model-free setting on prediction sets, see \cite{bartl2020pathwise}, \cite{bartl2019duality}, \cite{hou2018robust}, where prediction sets offer the opportunity to include beliefs in price developments or to select relevant price paths.

\appendix
\section{Neural Networks}
\label{sec:appendixNN}
For the reader's convenience we recall some results on neural networks. The following result essentially follows from \cite[Theorem~1]{hornik1991}. For completeness we include its proof here. 

\begin{theorem}\label{thm:universalapproxprob} Assume $\sigma$ is bounded and non-constant. Let $f \colon (\R^d,\cB(\bR^d)) \to (\R^m,\cB(\bR^m))$ be a measurable function and $\mu$ be a probability measure on $(\R^d,\mathcal{B}(\R^d))$. Then for any $\varepsilon, \tilde{\varepsilon} > 0$ there exists a neural network $g$ such that 
	\[ 
	\mu(\{x \in \R^d \colon \|f(x)-g(x)\|>\tilde{\varepsilon}\}) < \varepsilon.
	\]
\end{theorem}

\begin{proof}
	Let $\varepsilon, \tilde{\varepsilon} > 0$  be given and let $C>0$ satisfy that
	\begin{equation}\label{eq:eps1}
	\mu(\{x \in \R^d \colon \|f(x)\|>C\}) < \frac{\varepsilon}{2}.
	\end{equation}
	Define $\tilde{f} = \mathds{1}_{\{x \in \R^d \colon \|f(x)\|\leq C\}} f$. Then $\tilde{f} \in L^1(\R^d,\mu)$ and hence \cite[Theorem~1]{hornik1991} shows that there exists a neural network $g$ with 
	\[
	\int_{\R^d} \|\tilde{f}(x)-g(x)\| \mu(d x) < \frac{\varepsilon \tilde{\varepsilon}}{4}.
	\]
	Markov's inequality thus proves that 
	\begin{equation}\label{eq:eps2}
	\mu(\{x \in \R^d \colon \|\tilde{f}(x)-g(x)\|>\frac{\tilde{\varepsilon}}{2}\}) \leq \frac{2}{\tilde{\varepsilon}}\int_{\R^d} \|\tilde{f}(x)-g(x)\| \mu(d x) < \frac{\varepsilon}{2}.
	\end{equation}
	Combining \eqref{eq:eps1} and \eqref{eq:eps2} and recalling $f-\tilde f=f\mathds{1}_{\{x\in\bR^d:\|f(x)\|> C\}}$ yields
	\[ \begin{aligned}
	\mu\left(\{x \in \R^d \colon \|f(x)-g(x)\|>\tilde{\varepsilon}\}\right) &  \leq 	\mu\left(\{x \in \R^d \colon \|f(x)-\tilde{f}(x)\|>\frac{\tilde{\varepsilon}}{2}\} \cup \{x \in \R^d \colon \|\tilde{f}(x)-g(x)\|>\frac{\tilde{\varepsilon}}{2}\}\right)
	\\ & < \mu\left(\{x \in \R^d \colon \|f(x)\|> C\} \right) + \frac{\varepsilon}{2} < \varepsilon.
	\end{aligned}\]
	
\end{proof}

{\small 
\bibliographystyle{acm}
\bibliography{dl_optimalstopping}
}
\end{document}